%% file: paper.tex
\documentclass[11pt,a4paper]{article}

\input{paper-macros}

\title{On girth and the parameterized complexity of token sliding\\ and token jumping\footnote{The first two authors are supported by ANR project GrR (ANR-18-CE40-0032). The third author is supported in part by the Slovenian Research Agency (research project N1-0102). The fifth author is supported by URB project ``A theory of change through the lens of reconfiguration''.}}

\author{
Valentin Bartier
\thanks{Univ. Grenoble Alpes, CNRS, Grenoble INP, G-SCOP, France, e-mail: valentin.bartier@grenoble-inp.fr}
\and
Nicolas Bousquet
\thanks{CNRS, LIRIS, Universit\'e de Lyon, Universit\'e Claude Bernard Lyon 1, Lyon, France, e-mail: nicolas.bousquet@univ-lyon1.fr}
\and
Cl\'ement Dallard
\thanks{FAMNIT, University of Primorska, Koper, Slovenia, e-mail: clement.dallard@famnit.upr.si}
\and
Kyle Lomer
\thanks{Department of Computer Science, American University of Beirut, Lebanon, e-mail: kjl00@mail.aub.edu}
\and
Amer E. Mouawad
\thanks{Department of Computer Science, American University of Beirut, Lebanon, e-mail: aa368@aub.edu.lb}
}

\date{}

\begin{document}
\maketitle
\thispagestyle{empty}

\input{paper-abstract}
\input{paper-intro}
\input{paper-prelim}
\input{paper-positive-1}

\input{paper-positive-2}
\input{paper-positive-3}
\input{paper-negative-1}
\input{paper-negative-2}

\bibliographystyle{plain}
\bibliography{paper-references}
\end{document}

%% file: paper-macros.tex
\usepackage[margin=1in]{geometry}
\usepackage{graphics,color}
\usepackage{epsfig}
\usepackage{amsfonts}
\usepackage{mathrsfs}
\usepackage{algorithm,algorithmic}
\usepackage{subfigure}
\usepackage{hyperref}
\hypersetup{
    colorlinks,
    citecolor=black,
    filecolor=black,
    linkcolor=black,
    urlcolor=black
}
\usepackage{complexity}
\usepackage{amssymb,amsthm}
\usepackage{mathtools}
\usepackage{amsmath}
\usepackage{xcolor}
\usepackage{graphicx}
\usepackage{caption}
\usepackage[T1]{fontenc}
\usepackage{tikz}
\usetikzlibrary{arrows.meta,patterns,ipe}
\usepackage{stmaryrd}
\usepackage[capitalize]{cleveref} %

\newcommand{\WO}{\textsf{W[1]}}
\newcommand{\Oh}{\mathcal{O}}

\newcommand{\TJ}{\textsf{TJ}}
\newcommand{\TS}{\textsf{TS}}

\newtheorem{theorem}{Theorem}
\newtheorem{observation}{Observation}
\newtheorem{question}{Question}
\newtheorem{proposition}{Proposition}
\newtheorem{lemma}{Lemma}
\newtheorem{corollary}{Corollary}

%% file: paper-abstract.tex
\begin{abstract}
In the {\sc Token Jumping} problem we are given a graph $G = (V,E)$ and two independent sets $S$ and $T$ of $G$, each of size $k \geq 1$.
The goal is to determine whether there exists a sequence of $k$-sized independent sets in $G$, $\langle S_0, S_1, \ldots, S_\ell \rangle$, such that for every $i$,
$|S_i| = k$, $S_i$ is an independent set, $S = S_0$, $S_\ell = T$, and $|S_i \Delta S_{i+1}| = 2$.
In other words, if we view each independent set as a collection of tokens placed on a subset
of the vertices of $G$, then the problem asks for a sequence of independent sets which transforms $S$ to $T$ by individual token jumps which
maintain the independence of the sets. This problem is known to be PSPACE-complete on very restricted graph classes, e.g., planar bounded degree graphs and graphs of
bounded bandwidth. A closely related problem is the {\sc Token Sliding} problem, where instead of allowing a token to jump to any vertex of the graph we instead
require that a token slides along an edge of the graph. {\sc Token Sliding} is also known to be PSPACE-complete on the aforementioned graph classes.
We investigate the parameterized complexity of both problems on several graph classes, focusing on the effect of excluding certain cycles from the input graph.
In particular, we show that both {\sc Token Sliding} and {\sc Token Jumping} are fixed-parameter tractable on $C_4$-free bipartite graphs when parameterized by $k$.
For {\sc Token Jumping}, we in fact show that the problem admits a polynomial kernel on $\{C_3,C_4\}$-free graphs. In the case of {\sc Token Sliding}, we also
show that the problem admits a polynomial kernel on bipartite graphs of bounded degree. We believe both of these results to be of independent interest.
We complement these positive results by showing that, for any constant $p \geq 4$, both problems are W[1]-hard on $\{C_4, \dots, C_p\}$-free graphs
and {\sc Token Sliding} remains W[1]-hard even on bipartite graphs.
\end{abstract}

%% file: paper-intro.tex
\section{Introduction}\label{sec-intro}
Many algorithmic questions present themselves in the following form:
given the description of a system state and the description of a state we would ``prefer'' the system to be in, is it possible to transform the system
from its current state into the more desired one without ``breaking'' the system in the process? Such questions, with some generalizations and specializations, have
received a substantial amount of attention under the so-called \emph{combinatorial reconfiguration framework}~\cite{DBLP:journals/tcs/BrewsterMMN16,H13,Wrochna15}.
Historically, the study of reconfiguration questions predates the field of computer science, as
many classic one-player games can be formulated as reachability questions~\cite{JS79,KPS08}, e.g., the $15$-puzzle and Rubik's cube.
More recently, reconfiguration problems have emerged from computational problems in
different areas such as graph theory~\cite{CHJ08,IDHPSUU11,IKD12}, constraint satisfaction~\cite{GKMP09,DBLP:conf/icalp/MouawadNPR15},
computational geometry~\cite{DBLP:journals/comgeo/LubiwP15}, and
even quantum complexity theory~\cite{DBLP:conf/icalp/GharibianS15}. We refer the reader to the surveys by van den Heuvel~\cite{H13}
and Nishimura~\cite{DBLP:journals/algorithms/Nishimura18} for more background on combinatorial reconfiguration.

\paragraph*{Independent Set Reconfiguration.} In this work, we focus on the reconfiguration of independent sets.
Given a simple undirected graph $G$, a set of vertices $S \subseteq V(G)$ is an \emph{independent set} if the vertices
of this set are all pairwise non-adjacent. Finding an independent set of maximum cardinality, i.e., the {\sc Independent Set} problem,
is a fundamental problem in algorithmic graph theory and is known to be not only \NP-hard, but also \WO-hard and
not approximable within $\Oh(n^{1-\epsilon})$, for any $\epsilon > 0$, unless $\P = \NP$~\cite{DBLP:journals/toc/Zuckerman07}.
Moreover, {\sc Independent Set} is known to remain \WO-hard on graphs excluding $C_4$ (the cycle on four vertices) as an induced subgraph~\cite{DBLP:conf/iwpec/BonnetBCTW18}.

We view an independent set as a collection of tokens placed on the vertices of a graph such that no two tokens are adjacent.
This gives rise to (at least) two natural adjacency relations between independent sets (or token configurations), also called \emph{reconfiguration steps}.
These two reconfiguration steps, in turn, give rise to two combinatorial reconfiguration problems.
In the {\sc Token Jumping} (TJ) problem, introduced by Kami\'{n}ski et al.~\cite{KMM12}, a single reconfiguration step consists of first removing a token on some vertex $u$ and then immediately adding it back on any other vertex $v$, as long as no two tokens become adjacent. The token is said to \emph{jump} from vertex $u$ to vertex $v$.
In the {\sc Token Sliding} (TS) problem, introduced by Hearn and Demaine~\cite{DBLP:journals/tcs/HearnD05}, two independent sets are adjacent if one can be obtained from the other by a token jump from vertex $u$ to vertex $v$ with the additional requirement of $uv$ being an edge of the graph. The token is then said to \emph{slide} from vertex $u$ to vertex $v$ along the edge $uv$. Note that, in both the TJ and TS problems, the size of independent sets is fixed.
Generally speaking, in the {\sc Token Jumping} and {\sc Token Sliding} problems,
we are given a graph $G$ and two independent sets $S$ and $T$ of $G$. The goal is to determine
whether there exists a sequence of reconfiguration steps -- a \emph{reconfiguration sequence} -- that
transforms $S$ into $T$ (where the reconfiguration step depends on the problem).

Both problems have been extensively studied under the combinatorial reconfiguration framework, albeit under
different names~\cite{DBLP:journals/corr/BonamyB16,DBLP:conf/swat/BonsmaKW14,MR3295823,DBLP:conf/isaac/Fox-EpsteinHOU15,DBLP:conf/tamc/ItoKOSUY14,DBLP:journals/ieicet/ItoNZ16,KMM12,DBLP:conf/wads/LokshtanovMPRS15,DBLP:conf/isaac/MouawadNR14}.
It is known that both problems are \PSPACE-complete, even on restricted graph classes such as
graphs of bounded bandwidth (and then pathwidth)~\cite{DBLP:journals/corr/Wrochna14} and planar graphs~\cite{DBLP:journals/tcs/HearnD05}. In general {\sc Token Sliding} is more complicated to decide than {\sc Token Jumping}. However {\sc Token Sliding} and {\sc Token Jumping} can be decided in polynomial time on trees~\cite{MR3295823}, interval graphs~\cite{DBLP:journals/corr/BonamyB16}, bipartite permutation and bipartite distance-hereditary graphs~\cite{DBLP:conf/isaac/Fox-EpsteinHOU15}  or line graphs~\cite{IDHPSUU11}.
Lokshtanov and Mouawad~\cite{DBLP:journals/talg/LokshtanovM19} showed that, in bipartite graphs, {\sc Token Jumping} is \NP-complete
while {\sc Token Sliding} remains \PSPACE-complete. In split graphs, {\sc Token Jumping} is a trivial problem while {\sc Token Sliding} is PSPACE-complete~\cite{DBLP:conf/stacs/Belmonte0LMOS19}. In addition to the classes above, {\sc Token Jumping} can be decided in polynomial
time for even-hole-free graphs~\cite{KaminskiMM12}.

\begin{table}
\centering
\renewcommand{\arraystretch}{1.75}
\begin{tabular}{||c c c||}
 \hline
 Graph Class & \textsc{Token Jumping} & \textsc{Token Sliding} \\
 \hline\hline
 $\{C_3,C_4\}$-free graphs & FPT (Section~\ref{sec:pos0}) & Open \\
 \hline
 $C_4$-free graphs & W[1]-hard (Section~\ref{sec:neg1}) & W[1]-hard (Section~\ref{sec:neg1})  \\
 \hline
 Bipartite graphs & Open & W[1]-hard (Section~\ref{sec:neg2})  \\
 \hline
 Bipartite $C_4$-free graphs & FPT (Section~\ref{sec:pos0}) & FPT (Section~\ref{sec:pos3})  \\
 \hline
\end{tabular}
\caption{Parameterized complexity of {\sc Token Jumping} and {\sc Token Sliding} on several graph classes.}
\label{tab:complexity}
\end{table}

In this paper we focus on the parameterized complexity of the {\sc Token Jumping} and {\sc Token Sliding} problems on graphs where some cycles with prescribed length are forbidden.
Given an NP-hard problem, parameterized complexity permits to refine the notion of hardness: does it come from the whole instance or from a small parameter?
A problem $\Pi$ is \emph{FPT} (Fixed Parameterized Tractable) parameterized by $k$ if one can solve it in time $f(k) \cdot poly(n)$. In other words, the combinatorial
explosion can be restricted to a parameter $k$. In the rest of the paper, our parameter $k$ will be the size of the independent set (i.e.\ number of tokens).

Both {\sc Token Jumping} and {\sc Token Sliding} are known to be $W[1]$-hard\footnote{Informally, it means that they are very unlikely to admit an FPT algorithm.}
parameterized by $k$ on general graphs~\cite{DBLP:conf/wads/LokshtanovMPRS15}.
On the positive side, Lokshtanov et al. showed~\cite{DBLP:conf/wads/LokshtanovMPRS15} that {\sc Token Jumping} is FPT on bounded degree graphs.
{\sc Token Jumping} is also known to be FPT on strongly $K_{\ell,\ell}$-free graphs~\cite{IKO14,BousquetMP17}, a graph
being strongly $K_{\ell,\ell}$-free if it does not contain any $K_{\ell,\ell}$ as a subgraph.

\paragraph{Our result.} (For a complete overview of our results, see Table~\ref{tab:complexity}). In this paper, we focus on what happens if we consider
graphs that do not admit a (finite or infinite) collection of cycles of prescribed lengths. Such graph classes contain bipartite graphs (odd-hole-free graphs),
even-hole-free graphs and triangle-free graphs. Our main goal was to understand which cycles make the independent set reconfiguration problems hard.
Our main technical result consists in showing that {\sc Token Sliding} is $W[1]$-hard paramerized by $k$ on bipartite graphs with a reduction from
{\sc Multicolored Independent Set}. We were not able to adapt our reduction for {\sc Token Jumping} and left it as an open question:

\begin{question}\label{q1}
 \sloppy%
 Is {\sc Token Jumping} FPT parameterized by $k$ on bipartite graphs?
\end{question}

On the positive side, we prove that {\sc Token Jumping} admits a quadratic kernel (i.e.\ an equivalent instance of size $O(k^2)$ can be found in polynomial time) for $\{ C_3,C_4 \}$-free graphs while it is W[1]-hard if we restrict to $\{ C_4,\ldots,C_p \}$-free graphs for a fixed constant $p$ (the same hardness result also holds for {\sc Token Sliding}). Note that the fact that the problem is FPT on graphs of girth\footnote{The \emph{girth} of a graph is the length of its shortest cycle.} at least $5$ graphs also follows from FPT algorithms for strongly $K_{3,\ell}$-free graphs of~\cite{IKO14}, but even if a polynomial kernel can be derived from their result, the degree of our polynomial is better.
We were no able to remove the $C_4$ condition in order to obtain a parameterized algorithm for triangle-free graphs. If an FPT algorithm exists for triangle-free graphs, it would, in particular answer Question~\ref{q1}.

\begin{question}
  Is {\sc Token Jumping} FPT parameterized by $k$ on triangle-free graphs?
\end{question}

We then focus on {\sc Token Sliding}. While FPT algorithms are (relatively) easy to design on sparse graphs for {\sc Token Jumping}, they are much harder for {\sc Token Sliding}.
In particular, it is still open to determine if {\sc Token Sliding} is FPT on planar graphs or $H$-minor free graphs while they follow for instance
from~\cite{IKO14,BousquetMP17} for {\sc Token Jumping}. Our main positive result is that {\sc Token Sliding} on bipartite $C_4$-free graphs
(i.e.\ bipartite graphs of girth at least $6$) admits a polynomial kernel. Our proof is in two parts, first we show that {\sc Token Sliding} on bipartite graphs with
bounded degree admits a polynomial kernel and then show that, if the graphs admits a vertex of large enough degree then the answer is always positive.
So {\sc Token Sliding} is W[1]-hard on bipartite graphs but FPT on bipartite $C_4$-free graphs.
In our positive results, $C_4$-freeness really plays an important role (neighborhoods of the neighbors of a vertex $x$ are almost disjoint).
It would be interesting to know if forbidding $C_4$ is really important or whether it is only helpful with our proof techniques.
In particular, does {\sc Token Sliding} admit an FPT algorithm on bipartite $C_{2p}$-free graphs for some $p \ge 3$?
In our hardness reduction for bipartite graphs, all (even) cycles can appear and then such a result can hold.
Recall that we prove that {\sc Token Jumping} admits a polynomial kernel for graphs of girth at least $6$.
It would be interesting to see if our result on bipartite $C_4$-free graphs can be extended to this class.

\begin{question}
  Is {\sc Token Sliding} FPT parameterized by $k$ on graphs of girth at least $5$? Or, slightly weaker, is it FPT on graphs of girth at least $p$, for some constant $p$.
\end{question}

Note that the fact that the girth is at least $5$ is needed since {\sc Token Sliding} is W[1]-hard on bipartite graphs (which have girth at least $4$).
Let us finally briefly discuss some cases where we forbid an infinite number of cycles. We have already discussed the case where odd cycles are forbidden.
One can wonder what happens if even cycles are forbidden. It is shown in~\cite{KaminskiMM12} that {\sc Token Jumping} can be decided in polynomial time
for even-hole-free graphs (which is remarkable since computing a maximum independent set in this class is open).
However, as far as we know, the complexity status of the problem is open for {\sc Token Sliding}.
More generally, one can wonder what happens when we forbid all the cycles of length $p$ mod $q$ for every pair of integers $p,q$.

%% file: paper-prelim.tex
\section{Preliminaries}\label{sec-prelim}
We denote the set of natural numbers by $\mathbb{N}$.
For $n \in \mathbb{N}$ we let $[n] = \{1, 2, \dots, n\}$.

\paragraph{Graphs.}
We assume that each graph $G$ is finite, simple, and undirected. We let $V(G)$ and $E(G)$ denote the vertex set and edge set of $G$, respectively.
The {\em open neighborhood} of a vertex $v$ is denoted by $N_G(v) = \{u \mid uv \in E(G)\}$ and the
{\em closed neighborhood} by $N_G[v] = N_G(v) \cup \{v\}$.
For a set of vertices $Q \subseteq V(G)$, we define $N_G(Q) = \{v \not\in Q \mid uv \in E(G), u \in Q\}$ and $N_G[Q] = N_G(Q) \cup Q$.
The subgraph of $G$ induced by $Q$ is denoted by $G[Q]$, where $G[Q]$ has vertex set
$Q$ and edge set $\{uv \in E(G) \mid u,v \in Q\}$. We let $G - Q = G[V(G) \setminus Q]$.

A {\em walk} of length $\ell$ from $v_0$ to $v_\ell$ in $G$ is a vertex sequence $v_0, \ldots, v_\ell$, such that
for all $i \in \{0, \ldots, \ell-1\}$, $v_iv_{i + 1} \in E(G)$.
It is a {\em path} if all vertices are distinct. It is a {\em cycle}
if $\ell \geq 3$, $v_0 = v_\ell$, and $v_0, \ldots, v_{\ell - 1}$ is a path.
A path from vertex $u$ to vertex $v$ is also called a {\em $uv$-path}.
For a pair of vertices $u$ and $v$ in $V(G)$, by $\textsf{dist}_G(u,v)$ we denote the {\em distance} or length of a shortest $uv$-path
in $G$ (measured in number of edges and set to $\infty$ if $u$ and $v$ belong to different connected components).
The {\em eccentricity} of a vertex $v \in V(G)$, $\textsf{ecc}(v)$, is equal to $\max_{u \in V(G)}(\textsf{dist}_G(u,v))$.
The {\em radius} of $G$, $\textsf{rad}(G)$, is equal to $\min_{v \in V(G)}(\textsf{ecc}(v))$.
The {\em diameter} of $G$, $\textsf{diam}(G)$, is equal to $\max_{v \in V(G)}(\textsf{ecc}(v))$.
For $r \geq 0$, the {\em $r$-neighborhood} of a vertex $v \in V(G)$ is defined as
$N^r_G[v] = \{u \mid dist_G(u, v) = r\}$. We write
$B(v, r) = \{u \mid dist_G(u, v) \leq r\}$ and call it a {\em ball of radius $r$ around $v$};
for $S \subseteq V(G)$, $B(S, r) = \bigcup_{v \in S} B(v, r)$.

A graph $G$ is {\em bipartite} if the vertex set of $G$ can be partitioned into two disjoint sets $L$ (the left part) and $R$ (the right part), i.e. $V(G) = L \cup R$,
where $G[L]$ and $G[R]$ are edgeless.
Given two graphs $G$ and $H$, we say that
$G$ is \emph{$H$-free} if $G$ does not contain $H$ as an induced subgraph.

\paragraph{Reconfiguration.}
In the {\sc Token Jumping} problem we are given a graph $G = (V,E)$ and two independent sets $S$ and $T$ of $G$, each of size $k \geq 1$.
The goal is to determine whether there exists a sequence of $k$-sized independent sets in $G$, $\langle S_0, S_1, \ldots, S_\ell \rangle$, such that
$|S_i| = k$, $S_i$ is an independent set ($\forall i$), $S = S_0$, $S_\ell = T$, and $|S_i \Delta S_{i+1}| = 2$.
In other words, if we view each independent set as a collection of tokens placed on a subset
of the vertices of $G$, then the problem asks for a sequence of independent sets which transforms $S$ to $T$ by individual token jumps which
maintain the independence of the sets. For two independent sets $S$ and $T$, we write $S \leftrightsquigarrow T$ in $G$ if there exists a sequence
of jumps that transforms $S$ to $T$ in $G$. For the closely related problem of {\sc Token Sliding}, instead of allowing a token to jump to any vertex of the
graph we instead  require that a token slides along an edge of the graph. We use the same terminology for both problems as it will be clear from context which
problem we are referring to. Note that both {\sc Token Jumping} and {\sc Token Sliding} can be expressed in terms of a \emph{reconfiguration graph} $\mathcal{R}_Q(G,k)$,
where $Q \in \{\TS, \TJ\}$. Both $\mathcal{R}_{\TJ}(G,k)$ and $\mathcal{R}_{\TS}(G,k)$ contain a node for each independent set of $G$ of size exactly $k$.
We add an edge between two nodes whenever the independent set corresponding to one node can be obtained from the other by a single reconfiguration step. That is,
a single token jump corresponds to an edge in $\mathcal{R}_{\TJ}(G,k)$ and a single token slide corresponds to an edge in $\mathcal{R}_{\TS}(G,k)$.
Given two nodes $S$ and $T$ in $\mathcal{R}_{\TJ}(G,k)$ ($\mathcal{R}_{\TS}(G,k)$), the {\sc Token Jumping} problem ({\sc Token Sliding} problem) asks whether $S$ and $T$
belong to the same connected component of $\mathcal{R}_{\TJ}(G,k)$ ($\mathcal{R}_{\TS}(G,k)$).

%% file: paper-positive-1.tex
\section{Positive results}\label{sec-positive}

\subsection{{\sc Token Jumping} on $\{C_3,C_4\}$-free graphs and bipartite $C_4$-free graphs}\label{sec:pos0}

We say that a class of graphs $\mathcal{G}_\varepsilon$ is \emph{$\varepsilon$-sparse}, for some $\varepsilon>0$, if for every graph
$G \in \mathcal{G}$ with $n$ vertices, the number of edges in $G$ is at most $n^{2-\varepsilon}$.
By extension, $G$ is said to be $\varepsilon$-sparse.
Given an instance $\mathcal{I} = (G,S,T,k)$ of {\sc Token Jumping}, let $H = G - N_G[S \cup T]$ and $J$ denote the graph induced by $N_G[S \cup T]$.
In the remainder of this section, we show that $\mathcal{I}$ is a yes-instance
whenever (at least) one of the following two conditions is true:

\begin{enumerate}
\item $H$ is $\varepsilon$-sparse and contains more than $k(2k)^{1/\varepsilon}$ vertices, or
\item $J$ is $\{C_3,C_4\}$-free and contains a vertex of degree at least $3k$.
\end{enumerate}

\begin{lemma}\label{lem-large-sparse}
Let $\mathcal{I} = (G,S,T,k)$ be an instance of {\sc Token Jumping} and let $H = G - N_G[S \cup T]$.
If $H$ is an $\varepsilon$-sparse graph with more than $k(2k)^{1/\varepsilon}$ vertices then $\mathcal{I}$ is a yes-instance.
Moreover, the length of the shortest reconfiguration sequence from $S$ to $T$ is at most $2k$.
\end{lemma}

\begin{proof}
First, consider an $\varepsilon$-sparse graph $H'$ with $n > (2k)^{1/\varepsilon}$ vertices.
We claim that $H'$ contains a vertex with degree less than $\frac{n}{k}$.
Assume otherwise, i.e., suppose that the minimum degree in $H'$ is at least $\frac{n}{k}$.
Then, $|E(H')| \geq \frac{n^2}{2k}$. Moreover, since $H'$ is $\varepsilon$-sparse, it holds that $|E(H')| \leq n^{2-\varepsilon}$.
However, $\frac{n^2}{2k} \leq n^{2-\varepsilon}$ if and only if $n \leq (2k)^{1/\varepsilon}$, a contradiction.

Now, we shall prove, by induction on $k$, that $H$ contains an independent set of size at least $k$.
The statement holds for $k = 1$ (since $H$ must contain at least one vertex).
Now, consider the case where $k > 1$ and let $z$ be a vertex with minimum degree in $H$.
Following the above claim, $z$ has degree less than $\frac{n}{k}$.
Note that the graph $H' = H - N[z]$ contains at least $(k - 1) \frac{n}{k} \geq (k - 1) \frac{k(2k)^{1/\varepsilon}}{k} =  (k - 1) (2k)^{1/\varepsilon}$ vertices.
By the induction hypothesis, $H'$ contains an independent set $X$ of size at least $k-1$.
Thus, $X \cup \{z\}$ is an independent set in $H$ of size at least $k$.

Hence, we can tranform $S$ to $T$ by simply jumping all the tokens in $S$ to an independent set $X \subseteq V(G) \setminus (S \cup T)$ and then
from $X$ we jump the tokens (one by one) to $T$. This completes the proof.
\end{proof}

\begin{lemma}\label{lem-neighborhood-large-deg}
Let $\mathcal{I} = (G,S,T,k)$ be an instance of {\sc Token Jumping} and let $J$ denote the graph induced by $N_G[S \cup T]$.
If $J$ is $\{C_3,C_4\}$-free and contains a vertex $v$ of degree at least $3k$, then $\mathcal{I}$ is a yes-instance.
Moreover, the length of the shortest reconfiguration sequence from $S$ to $T$ is at most $2k$.
\end{lemma}

\begin{proof}
Fix $w \in S \cup T$.
First, observe that for any $u \in N(S \cup T)$, $u$ is either ajdacent to $w$ and no neighbor of $w$ (otherwise $J$ would contain a $C_3$), or $u$ is adjacent to at most one neighbor of $w$ (otherwise $J$ would contain a $C_4$).
Therefore, every vertex in $N(S \cup T)$ has degree at most $2k$.
As $J$ is $C_3$-free, $N_J(w)$ is an independent set.
Furthermore, for any $u,v \in N_J(w)$, $u \neq v$, we have $N_J(u) \cap N_J(v) = \{w\}$, that is, $w$ is the only common neighbor of $u$ and $v$ in $J$; otherwise, $J$ would contain $C_4$.
Hence, if $w$ has at least $3k$ neighbors, then at least $k$ of them only have $w$ as a neighbor in $S \cup T$.
Thus, we can jump the tokens on $S$ to $N(w)$, starting with the token on $w$, if any.
Then, we can jump the tokens on the vertices in $T$.
Clearly, the length of such a reconfiguration sequence is at most $2k$.
\end{proof}

\begin{proposition}\label{thm-sparse}
Let $\mathcal{I} = (G,S,T,k)$ be an instance of {\sc Token Jumping}, let $H = G - N_G[S \cup T]$, and let $J$ denote the graph induced by $N_G[S \cup T]$.
If $H$ is $\varepsilon$-sparse, $\varepsilon > 0$, and $J$ is $\{C_3,C_4\}$-free then $\mathcal{I}$ admits a kernel with $\Oh(k^2 + k^{1 + 1/\varepsilon})$ vertices.
\end{proposition}

\begin{proof}
If $H$ contains more than $k(2k)^{1/\varepsilon}$ vertices then $\mathcal{I}$ is a yes-instance (by Lemma~\ref{lem-large-sparse}).
If $J$ contains a vertex of degree $3k$ or more then, again, $\mathcal{I}$ is a yes-instance (by Lemma~\ref{lem-neighborhood-large-deg}).
Putting it all together, we have $|S \cup T| \leq 2k$, $|N_G(S \cup T)| \leq 2k(3k-1) = 6k^2-2k = \Oh(k^2)$, and
$|V(G) \setminus N_G[S \cup T]| \leq k(2k)^{1/\varepsilon} = \Oh(k^{1 + 1/\varepsilon})$.
\end{proof}

\begin{theorem}\label{c3 c4 free}
{\sc Token Jumping} parameterized by $k$ admits a kernel with at most $\Oh(k^{2})$ vertices on $\{C_3,C_4\}$-free graphs as well as bipartite $C_4$-free graphs.
\end{theorem}

\begin{proof}
Let $\mathcal{I} = (G,S,T,k)$ be an instance of {\sc Token Jumping} such that $G$ is $\{C_3,C_4\}$-free.
Let $H = G - N_G[S \cup T]$ and $J$ denote the graph induced by $N_G[S \cup T]$.
Since $J$ is $\{C_3,C_4\}$-free, \cref{lem-neighborhood-large-deg} implies that if $J$ contains more than $6k^2 - 2k$ vertices, then $\mathcal{I}$ is a yes-instance.
Kim showed that a $C_3$-free graph with $\Oh(k^2 / (\log{k}))$ vertices contains an independent set of size at least $k$~\cite{MR1369063}.
Hence, if $H$ contains more than $\Oh(k^2 / (\log{k}))$ vertices, then $\mathcal{I}$ is a yes-instance.
Thus, $G$ contains at most $\Oh(k^{2})$ vertices.
The same result holds for bipartite $C_4$-free graphs since they are $\{C_3,C_4\}$-free.
\end{proof}

%% file: paper-positive-2.tex
\subsection{{\sc Token Sliding} on bounded-degree bipartite graphs}
Unlike the case of {\sc Token Jumping}, it is not known whether {\sc Token Sliding} is fixed-parameter tractable (parameterized by $k$) on graphs of bounded degree.
In this section we show that it is indeed the case for bounded-degree bipartite graphs. This result, interesting in its own right, will be crucial
for proving that {\sc Token Sliding} is fixed-parameter tractable on bipartite $C_4$-free graphs in the next section.
We start with a few definitions and needed lemmas.

Let $R(G,I) = \{v \mid v \in \bigcap_{I' \mid I \leftrightsquigarrow I'}{I'}\}$ be the subset of $I$ containing all of the tokens $v$ such that
$v \in I'$ for all $I'$ reachable from $I$. In other words, the tokens on vertices of $R(G,I)$ can never move in any reconfiguration sequence starting from $I$. We call vertices in $R(G,I)$ \emph{rigid with respect to $G$ and $I$}.
An independent set $I$ is said to be \emph{unlocked} if $R(G,I) = \emptyset$. Given a graph $G$ and $r \geq 1$, a set $S \subseteq V(G)$ is called an \emph{$r$-independent set}, or \emph{$r$-independent} for short, if $B(v, r) \cap S = \{v\}$, for all $v \in S$.
Note that a $1$-independent set is a standard independent set and
a $r$-independent set, $r > 1$, is a set where the shortest path
between any two vertices of the set contains at least $r$ vertices (excluding the endpoints).

For a vertex $v \in V(G)$ and a set $S \subseteq V(G) \setminus \{v\}$, we
let $D(v, S)$ denote the set of vertices in $S$ that are closest to $v$. That is, $D(v, S)$ is the set of vertices in $S$ whose distance to $v$ is minimum.
We say $D(v, S)$ is \emph{frozen} if $|D(v, S)| \geq 2$ and it
is not possible to slide a single token in $D(v, S)$ to obtain $S'$ such that either $v \in S'$ or $|D(v, S')| = 1$.
Note that, in time polynomial in $n = |V(G)|$, it can be verified whether $D(v, S)$ is frozen by simply checking, for each vertex $u \in D(v, S)$,
whether $u$ can slide to a vertex $w$ which is closer to $v$ (or to $v$ itself if $u$ is adjacent to $v$).

\begin{lemma}[\cite{DBLP:conf/isaac/Fox-EpsteinHOU15}]\label{lem-rigid-bipartite}
$S \leftrightsquigarrow T$ in $G$ if and only if $R(G,S) = R(G,T)$ and
$(S \setminus R(G,S)) \leftrightsquigarrow (T \setminus R(G,S))$ in $G - N[R(G,S)]$. Moreover,
if $G$ is bipartite then $R(G,S)$ and $R(G,T)$ can be computed in time linear in $|V(G)| = n$.
\end{lemma}

\begin{lemma}[\cite{DBLP:conf/isaac/Fox-EpsteinHOU15}]\label{lem-switch-bipartite}
Let $G = (L \cup R, E)$ be a bipartite graph and let $S$ be an unlocked
independent set of $G$. Then, in time linear in $n$, we can compute
a reconfiguration sequence $\langle S = I_0, I_1, \ldots, I_\ell \rangle$ where $I_\ell \cap L = \emptyset$ and $\ell = |S \cap L|$.
\end{lemma}

The next lemma was also proved in~\cite{DBLP:conf/isaac/Fox-EpsteinHOU15} but we repeat the proof here both for completeness and since
we will use similar ideas in subsequent proofs.

\begin{lemma}[\cite{DBLP:conf/isaac/Fox-EpsteinHOU15}]\label{lem-move-bipartite}
Let $G = (L \cup R, E)$ be a connected bipartite graph and let $S$ be an unlocked
independent set of $G$. Let $v \in V(G) \setminus S$ and let $D(v, S) \subseteq L$ (or symmetrically $D(v, S) \subseteq R$).
Then, in time linear in $|V(G)| = n$, one can find a reconfiguration sequence $\langle S = I_0, I_1, \ldots, I_\ell \rangle$ where $v \in I_\ell$
and $\ell$ is at most $|S \cap L| - 1$ (or symmetrically $|S \cap R| - 1$) plus the distance between $v$ and a token of $D(v, S)$.
\end{lemma}

\begin{proof}
There are two cases to consider:\newline
(1) If there is a unique token $u \in D(v, S) \subseteq S$ which is closest to $v$ then the reconfiguration sequence
is constructed by repeatedly moving the token on $u$ to a vertex which is closer to $v$.
Let $w$ be any vertex in $N(u)$ where some shortest path from $u$ to $v$ passes through $w$.
Since $u$ is uniquely closest to $v$, it must be the case that $N(w) \cap S = \{u\}$.
Hence, we construct $I_1 = (I \setminus \{u\}) \cup \{w\}$; as $w$ is now uniquely closest to $v$ the process can be iterated.
The same strategy can be applied if $D(v, S)$ is not frozen.\newline
(2) Assume $D(v, S)$ is frozen. Let $d$ denote the distance from $v$ to any vertex $u \in D(v, S)$. Without loss of generality, we can assume
that $D(v, S) \subseteq L$ (the other case is symmetric). We apply Lemma~\ref{lem-switch-bipartite} which guarantees (in linear time) the existence
of a computable reconfiguration sequence $\langle S = I_0, I_1, \ldots, I_\ell \rangle$ where $I_\ell \cap L = \emptyset$ and $\ell = |S \cap L|$.
There exists an index $j$, with $j < \ell < |S \cap L|$, where $I_j$ has a unique token $u$ which is closest to $v$. This follows from
the fact that some tokens of $D(v, S)$ will move to be at distance $d + 1$ from $v$ (possibly all but one)
leaving a vertex $u$ uniquely closest to $d$.
Therefore, we can now apply the same strategy as in the previous case. The reconfiguration sequence will
be of length at most $j + d$, as needed.
\end{proof}

Let $\mathcal{I} = (G = (V, E), S, T, k)$ be an instance of {\sc Token Sliding} where $G$ is a bipartite graph of bounded degree $\Delta$.
We assume, without loss of generality, that $G$ is connected;
as otherwise we can solve the problem independently on each component of $G$ (and there are at most $k$ components containing tokens).
Moreover, given Lemma~\ref{lem-rigid-bipartite}, we can assume, without loss of generality, that $S$ and $T$ are unlocked.
In other words, we assume that it has been verified that $R(G,S) = R(G,T)$ and $N[R(G,S)]$ has been deleted from $G$.
We now give a slightly different version of Lemma~\ref{lem-move-bipartite} better suited for our needs.

\begin{lemma}\label{lem-move-far}
Let $G$ be a connected bipartite graph and let $S$ be an unlocked
independent set of $G$. Let $v$ be a vertex in $V(G) \setminus S$ such that $N_G[v] \cap S = \emptyset$.
Let $D(v, S) \subseteq L$ (or symmetrically $D(v, S) \subseteq R$) such that $\textsf{dist}_G(u,v) = d$, for all $u \in D(v, S)$.
Then, in time linear in $|V(G)| = n$, we can find a reconfiguration sequence $\langle S = I_0, I_1, \ldots, I_\ell \rangle$,
where $I_\ell = (S \setminus \{u\}) \cup \{v\}$ for some vertex $u$ in $D(v, S)$
and $\ell$ is at most $2(|S| - 1) + d$.
\end{lemma}

\begin{proof}
Similarly to the proof of Lemma~\ref{lem-move-bipartite}, there are two cases to consider:\newline
(1) If there is a unique token $u \in D(v, S)$ which is closest to $v$ or $D(v, S)$ is not frozen
then the reconfiguration sequence obtained by repeatedly moving the token on $u$ to a vertex which is closer to $v$ gives us the required sequence.
Since no other token is moved, we have $I_\ell = (S \setminus \{u\}) \cup \{v\}$. \newline
(2) In the other case, we have $D(v, S) \geq 2$ and $D(v, S)$ is frozen. We assume, without loss of generality, that $D(v, S) \subseteq L$.
We apply Lemma~\ref{lem-switch-bipartite} which returns a reconfiguration sequence $\langle S = I_0, I_1, \ldots, I_\ell \rangle$
where $I_\ell \cap L = \emptyset$ and $\ell = |S \cap L|$.
There exists an index $j$, with $j < \ell < |S \cap L|$, where $I_j$ has a unique token $u \in D(v, S)$ which is closest to $v$.
Let $\alpha = \langle I_0, I_1, \ldots, I_j \rangle$. Note that $\alpha$ slides exactly $j$ distinct tokens (not including $u$) from $L$ to $R$.
We let $M_\alpha$ denote these tokens. Moreover, $\alpha$ is reversable. Hence, we let $\alpha^{-1}$ denote the sequence consisting
of applying the slides of $\alpha$ in reverse order.
Now, we construct a sequence $\beta$ of slides that moves the token on $u$ to $v$.
Recall that this is a sequence of exactly $d$ slides that repeatedly slides the same token.
We denote the resulting independent set (after applying $\alpha \cdot \beta$) by $I_\beta$.
We claim that $\gamma = \alpha \cdot \beta \cdot \alpha^{-1}$ is the required sequence that transforms $S$ to $(S \setminus \{u\}) \cup \{v\}$.
To see why $\gamma$ is a valid reconfiguration sequence, it suffices to show that $N_G[M_\alpha] \cap N_G[v] = \emptyset$.
Since $N_G[v] \cap S = \emptyset$, we know that $d \geq 2$ if both $v$ and $D(v, S)$ are contained in $L$ (or $R$) and $d \geq 3$ otherwise.
If $\{v\}, D(v, S) \subseteq L$ (or $\{v\}, D(v, S) \subseteq R$) then every vertex in $M_\alpha$ is at distance at least three from $v$, as needed.
Finally, if $v \in L$ and $D(v, S) \subseteq R$ (or $v \in R$ and $D(v, S) \subseteq L$) then every vertex in $M_\alpha$ is at distance at least four from $v$.
\end{proof}

\begin{lemma}\label{lem-move-2ind}
If $G$ is a connected graph and $S$ and $T$ are any two $2$-independent sets of $G$ such that $S \cup T$ is also $2$-independent then $S \leftrightsquigarrow T$ in $G$.
\end{lemma}

\begin{proof}
We proceed by induction on $|S \Delta T| = |(S \setminus T) \cup (T \setminus S)|$, i.e., the size of the symmetric difference between $S$ and $T$.
If $|S \Delta T| = 0$ then $S = T$ and there is nothing to prove. Hence, we assume that the statement is true for $|S \Delta T| = q > 0$. We compute
a shortest path between all pairs of vertices $(u,v)$ in $G$, where $u \in S\setminus T$ and $v \in T \setminus S$. We let $(u,v)$
denote a pair where the distance is minimized and we fix a shortest path between $u$ and $v$. There are two cases to consider:\newline
(1) If $S \cap T = \emptyset$ then we can simply slide $u$ to $v$ along the shortest path and we are done. To see why, recall
that both $S$ and $T$ are $2$-independent. Hence, they are both unlocked and if there is more than one vertex in $S \setminus T$ that
is closest to $v$ then we can simply slide $u$ into one of its neighbors, say $w$, that is closer to $v$ to obtain a unique vertex which is closest
to $v$; none of those neighbors are adjacent to a vertex in $S$ since $S$ is $2$-independent. Now, assume that there exists a vertex $x$
along the shortest path from $w$ to $v$ such that $x \in N(y)$, $y \in S$. This contradicts the choice of $u$ since $y$ is closer to $v$.\newline
(2) If $S \cap T \neq \emptyset$ then there are two cases. When the shortest path from $u$ to $v$ does not contain any vertex
in $N_G[S \cap T]$ then we apply the same reasoning as above. Otherwise, let $W = w_1, w_2, \ldots, w_q$ denote the vertices in
$N_G[S \cap T]$ along the shortest path from $u$ to $v$ (sorted in the order in which they are visited).
We divide $W$ into three sets $X = W \cap (S \cap T)$, $Y = W \cap (N_G(X))$, and $Z = W \setminus (X \cup Y)$.
In other words, $X$ denotes the set of vertices in $S \cap T$, $Y$ denotes the vertices used as entry and exit points for the vertices in $X$,
and $Z$ denotes the vertices in $N_G(S \cap T)$ visited along the shortest path without passing through a vertex of $N_G(Z) \cap (S \cap T)$.
Since $S \cap T$ is $2$-independent, no vertex in $Y \cup Z$ can have two neighbors in $S \cap T$.
Moreover, since we have a shortest path from $u$ to $v$, if there exists $x \in X$ then $N_G(x) \cap Z = \emptyset$. In particular,
the shortest path either visits a vertex $x \in S \cap T$ and two of its neighbors or only visits at most three neighbors of $x$;
as otherwise we can find a shorter path from $u$ to $v$. If the shortest path visits three neighbors $w$, $y$, and $z$, of a vertex $x \in S \cap T$ then we can
safely replace this sub-path by $w$, $x$, $z$. Hence, we assume in what follows that the shortest path visits at most two neighbors of any vertex in $S \cap T$.
We construct, from $W$, the sequence $A = a_1, a_2, \ldots, a_p$ of ``affected'' vertices in $S \cap T$. In other words,
if the shortest path from $u$ to $v$ visits a vertex in $S \cap T$ or visits one or two of its neighbors then we add the vertex to $A$ (in the order in which the visits occur).
We now proceed iteratively as follows. We slide $a_p$ to $v$, then $a_{p-1}$ to $a_p$, $\ldots$, and then finally we slide $u$ to $a_1$.
Note that between every one of those pairs of vertices we have a shortest path; since we are sliding along the shortest path from $u$ to $v$.
Moreover, after moving each token to its target position, we maintain a $2$-independent set $S'$. Therefore, for each such shortest path the intersection
with $N_G[S']$ remains empty.
\end{proof}

Let $G$ be a graph and let $X \subseteq V(G)$. The \emph{interior} of $X$ is
the set of vertices in $X$ at distance at least three from $V(G) \setminus X$ (separated by at least two vertices).
We say a set $X$ is \emph{fat} if its interior is connected and contains a $2$-independent set of size at least $2k$.

\begin{lemma}\label{lem-fat}
Let $G$ be a graph of maximum degree $\Delta$. Let $v \in V(G)$ and $r \in
\mathbb{N}$. If $B(v,r)$ contains more than $2k(1 + \Delta + \Delta^2)^2$ vertices then $B(v,r)$ is fat.
\end{lemma}

\begin{proof}
We only need to prove that the interior of $B(v,r)$, that is $B(v, r-2)$, contains a $2$-independent set of size at least $2k$; as $B(v, r - 2)$ is connected by construction.
First, note that any graph of maximum degree $\Delta$ on more than $2k(1 + \Delta + \Delta^2)$ vertices must contain a $2$-independent set of size at least $2k$.
So it suffices to show that $B(v, r-2)$ contains more than $2k(1 + \Delta + \Delta^2)$ vertices.
We divide $B(v,r)$ into layers, where $L_0 = \{v\}$, $L_1 = N(v)$, $\ldots$, and $L_r = N^r(v)$.
Since $G$ has maximum degree $\Delta$, for every $i \geq 1$, layer $L_i$ contains at most $(\Delta-1)^{i-1}\Delta$ vertices.
If $B(v,r-2)$ contains more than $2k(1 + \Delta + \Delta^2)$ vertices then we are done.
Otherwise, $L_{r-2}$ must contain at most $2k(1 + \Delta + \Delta^2)$ vertices. Consequently,
$L_{r-1} \cup L_r$ would contain at most $2k\Delta(1 + \Delta + \Delta^2) + 2k\Delta^2(1 + \Delta + \Delta^2) = (1 + \Delta + \Delta^2)(2k\Delta+2k\Delta^2)$ vertices.
Therefore, $B(v, r)$ contains at most $2k(1 + \Delta + \Delta^2) + (1 + \Delta + \Delta^2)(2k\Delta+2k\Delta^2) = (1 + \Delta + \Delta^2)(2k + 2k\Delta+2k\Delta^2)$
which is equal to $2k(1 + \Delta + \Delta^2)^2$ vertices, a contradiction.
\end{proof}

\begin{lemma}\label{lem-move-fat}
Let $\mathcal{I} = (G, S, T, k)$ be an instance of {\sc Token Sliding} where $G$ is a bounded-degree bipartite graph.
If $V(G) \setminus (S \cup T)$ contains a fat set $X$ then $\mathcal{I}$ is a yes-instance.
\end{lemma}

\begin{proof}
First, recall that we assume that $G$ is connected and both $S$ and $T$ are unlocked.
Let $I$ be a $2$-independent set of size $2k$ in the interior of $X$ (at distance at least three from any vertex outside of $X$).
We prove that $S$ can be transformed into $S' \subset I$. Similar arguments hold for transforming $T$ into $T' \subset I$.
Hence, the statement of the theorem follows by applying Lemma~\ref{lem-move-2ind} on $S'$ and $T'$.

We proceed by induction on $|S \Delta S'|$, i.e., the size of the symmetric difference between $S$ and $S'$.
If $|S \Delta S'| = 0$ then $S = S'$ and we are done. Otherwise, we reduce the size of the symmetric difference as follows.
Recall that initially $S \cap S' = \emptyset$; as $X \subseteq V(G) \setminus (S \cup T)$.
However, the size of the intersection will increase as more tokens are moved to $S'$.
We pick a pair $(u,v)$ such that $u \in S \setminus S'$ and $v \in S'$ and the distance between $u$ and $v$ is minimized.
There are two cases 2 consider: \newline
(1) If $v$ does not contain a token (or $v \in S' \setminus S$) then the shortest path from $u$ to $v$ does not intersect with $N_G[S' \cap S]$. We therefore
invoke Lemma~\ref{lem-move-far} in the graph $G - (N[S' \cap S])$. This guarantees that the token on $u$ slides to $v$ and every other token remains in place.\newline
(2) Otherwise, $v$ already contains a token (or $v \in S' \cap S$). We invoke Lemma~\ref{lem-move-2ind} on the graph
induced by the interior of $X$ and transform $C = S' \cap S \subset I$ into another $2$-independent set $C' \subseteq I$ that does not contain $v$; this is possible since $|C| = |C'| \leq k$.
Now we can again invoke Lemma~\ref{lem-move-far} similarly to the previous case.
\end{proof}

\begin{theorem}\label{thm-bounded-degree-bipartite}
\sloppy%
{\sc Token Sliding} parameterized by $k$ admits a kernel with $\Oh(k^2\Delta^5)$ vertices on bounded-degree bipartite graphs.
Moreover, the problem can be solved in $\Oh^\star(k^{2k}\Delta^{5k})$-time.
\end{theorem}

\begin{proof}
Let $\mathcal{I} = (G, S, T, k)$ be an instance of {\sc Token Sliding} where $G$ is a bipartite graph of maximum degree $\Delta$.
We assume, without loss of generality, that $G$ is connected and $S$ and $T$ are unlocked; for otherwise we can solve connected components independently and
we can return a trivial no-instance if $R(G,S) \neq R(G,T)$ (Lemma~\ref{lem-rigid-bipartite}).
Next, from Lemmas~\ref{lem-fat} and~\ref{lem-move-fat}, we know that each connected component of $V(G) \setminus (S \cup T)$
contains at most $\Oh(k\Delta^4)$ vertices; otherwise we can return a trivial yes-instance.
Since the number of components in $V(G) \setminus (S \cup T)$ is bounded by $2k\Delta$ and $|S \cup T| \leq 2k$, we get the desired bound.
To solve the problem, it suffices to construct the complete reconfiguration graph and verify if $S$ and $T$ belong to the same connected component. This concludes the proof.
\end{proof}

%% file: paper-positive-3.tex
\subsection{{\sc Token Sliding} on bipartite $C_4$-free graphs}\label{sec:pos3}
Equipped with Theorem~\ref{thm-bounded-degree-bipartite}, we are now ready to prove that {\sc Token Sliding} admits a polynomial kernel on bipartite $C_4$-free graphs.
Our strategy will be simple. We show that if the graph contains a vertex of large degree then we have a yes-instance. Otherwise,
we invoke Theorem~\ref{thm-bounded-degree-bipartite} to obtain the required kernel.

We start with a few simplifying assumptions. Let $\mathcal{I} = (G, S, T, k)$ be an instance
of {\sc Token Sliding} where $G = (L \cup R, E)$ is a connected bipartite $C_4$-free graphs.
We assume that both $S$ and $T$ are unlocked (Lemma~\ref{lem-rigid-bipartite}).
Moreover, we assume that each vertex in $G$ can have at most one pendant neighbor. This assumption is safe because
no two tokens can occupy two pendant neighbors of a vertex; as otherwise $S$ or $T$ would be locked. Moreover, if a token is placed on
a pendant neighbor of a vertex $v$ then no other token can reach $v$.

Let $v \in V(G)$ be a vertex of degree at least $k^{2} + k + 1$ in $G$.  We let $u_p$ denote the pendant neighbor of $v$ (if it exists).
We assume, without loss of generality, that $v \in L$. We let $N_1 = N_G(v) \setminus \{u_p\} = \{u_1, u_2, \ldots, u_q\}$, $N_2 = N^2_G(v)$, and $N_3 = N^3_G(v)$.
Since $G$ is bipartite, $N_1 \subseteq R$, $N_2 \subseteq L$, and $N_3 \subseteq R$.
Moreover, since $G$ is $C_4$-free, each vertex in $N_2$ has exactly one neighbor in $N_1$.
Therefore, we partition $N_2$ into sets $N_{u_1}$, $N_{u_2}$, $\ldots$, $N_{u_q}$, where each set $N_{u_i}$ contains
the neighbors of $u_i$ in $N_2$, that is, $N(u_i) \setminus \{v\}$.
We also partition $N_3$ into two sets $M_{\text{small}}$ and $M_{\text{big}}$.
Each vertex in $M_{\text{big}}$ contains vertices connected to at least $k + 1$ sets in $N_2$. Note that, because of $C_4$-freeness, each vertex in $N_3$ is connected
to at most one vertex of any set $N_{u_i}$. We let $M_{\text{small}} = N_3 \setminus M_{\text{big}}$. Each vertex in $M_{\text{small}}$ has at most $k$ neighbors in $N_2$.
In other words, each vertex in $M_{\text{small}}$ is connected to at most $k$ sets, each one of those sets being the neighborhood of a vertex in $N_1$.

We now proceed in five stages. We first show how to transform $S$ to $S_1$ such that $S_1 \cap B(v,3) \subseteq N_2$. In other words,
we can guarantee that all tokens in the ball of radius three around $v$ are contained in $N_2$.
We then tranform $S_1$ to $S_2$ such that $S_2 \cap B(v,3) \subseteq N_1 \cup N_3$. Next, we
tranform $S_2$ to $S_3$ such that $S_3 \cap B(v,3) \subseteq N_1 \cup M_{\text{small}}$.
Then, we tranform $S_3$ to $S_4$ such that $S_4 \cap B(v,3) \subseteq N_1$ and finally to $S_5$ such that $S_5 \subseteq N_1$.
By applying the same strategy starting from $T$, we obtain $T_5 \subseteq N_1$.
We conclude our proof by showing that $S_5$ can be transformed to $T_5$.

\begin{lemma}\label{lem-move-highdeg1}
Let $G$ be a connected bipartite $C_4$-free graphs and let $v \in V(G)$ be a vertex of degree at least $k^{2} + k + 1$.
Let $S$ be an unlocked independent set of $G$ of size $k$. Then, there exists $S'$ such that $S \leftrightsquigarrow S'$ and $S' \cap B(v,3) \subseteq N_2$.
\end{lemma}

\begin{proof}
We invoke Lemma~\ref{lem-switch-bipartite} and move all tokens in $R$ to $L$ (since $S$ is unlocked). We denote the resulting set by $S'$.
Consequently, we know that $S' \cap B(v,3) \subseteq L$. If there is no token on $v$ then we are done; as $v \in L$, $N_1 \subseteq R$, $N_2 \subseteq L$, and $N_3 \subseteq R$.
Otherwise, given that $v$ has degree at least $k + 1$, there must exist at least one path $P = v, x, y$ such that $N_G[P] \cap S' = \{v\}$.
Hence, we can slide the token on $v$ to $y$. This completes the proof.
\end{proof}

\begin{lemma}\label{lem-move-highdeg2}
Let $G$ be a connected bipartite $C_4$-free graphs and let $v \in V(G)$ be a vertex of degree at least $k^{2} + k + 1$.
Let $S$ be an unlocked independent set of $G$ of size $k$ such that $S \cap B(v,3) \subseteq N_2$.
Then, there exists $S'$ such that $S \leftrightsquigarrow S'$ and $S' \cap B(v,3) \subseteq N_1 \cup N_3$.
\end{lemma}

\begin{proof}
Since $S \cap B(v,3) \subseteq N_2$, we simply have to invoke Lemma~\ref{lem-switch-bipartite} and move all tokens in $L$ to $R$.
Note that no token can reach $u_p$ in a single slide, as needed.
\end{proof}

\begin{lemma}\label{lem-move-highdeg3}
Let $G$ be a connected bipartite $C_4$-free graphs and let $v \in V(G)$ be a vertex of degree at least $k^{2} + k + 1$.
Let $S$ be an unlocked independent set of $G$ of size $k$ such that $S \cap B(v,3) \subseteq N_1 \cup N_3$.
Then, there exists $S'$ such that $S \leftrightsquigarrow S'$ and $S' \cap B(v,3) \subseteq N_1 \cup M_{\text{small}}$.
\end{lemma}

\begin{proof}
We make use of the fact that each vertex in $M_{\text{big}}$ is connected to at least $k + 1$
sets in $N_2$ and hence is connected (via a vertex in $N_2$) to at least $k + 1$ vertices in $N_1$.
Let $w$ be a vertex in $S \cap M_{\text{big}}$; if $S \cap M_{\text{big}}$ then $S \cap B(v,3) \subseteq N_1 \cup M_{\text{small}}$ and we are done.
Recall that $|S \cap N_1| + |S \cap N_3| \leq k$, no two vertices
in $N_3$ have two common neighbors in $N_2$, and no two vertices in $N_2$ have two common neighbors in $N_1$.
Hence, there exists at least $k + 1$ vertex-disjoint path connecting $v$ to $w$.
At least one such path, say $P = \{w, x, y, z, v\}$, satisfies $N_G[P] \cap S = \{w\}$.
We slide $w$ to $z$ and call the resulting set again $S$ for simplicity.
This process is repeated as long as there are tokens in $M_{\text{big}}$.
We let $S'$ denote the resulting set, i.e., where $S' \cap M_{\text{big}} = \emptyset$.
\end{proof}

\begin{lemma}\label{lem-move-highdeg4}
Let $G$ be a connected bipartite $C_4$-free graphs and let $v \in V(G)$ be a vertex of degree at least $k^{2} + k + 1$.
Let $S$ be an unlocked independent set of $G$ of size $k$ such that $S \cap B(v,3) \subseteq N_1 \cup M_{\text{small}}$.
Then, there exists $S'$ such that $S \leftrightsquigarrow S'$ and $S' \cap B(v,3) \subseteq N_1$.
\end{lemma}

\begin{proof}
Since $S \cap B(v,3) \subseteq N_1 \cup M_{\text{small}}$, we know that every token not in $N_1$ must be in $M_{\text{small}}$.
We let $A$ denote the subset of $M_{\text{small}}$ containing tokens. Note that if $A$ is empty, then we are done.
Otherwise, we know that each token in $A$ is connected to at most $k$ sets in $N_2$ (by construction) and therefore at most $k$ vertices in $N_1$.
We let $B$ denote the at most $k^2$ subsets of $N_2$ that contain a vertex with a neighbor in $A$. We let $C$ denote the at most $k^2$ vertices
of $N_1$ whose neighborhoods are in $B$. We proceed in two stages.
First we move all tokens in $C$ to some vertex in $N_1 \setminus C$.
To do so, we invoke Lemma~\ref{lem-move-far} as follows.
If there are any tokens originally in $N_1 \setminus C$, then we move them to one of their neighbors in $N_2$ (this is possible since no two vertices in $N_2$
have two common neighbors in $N_1$ and there are no tokens in $M_{\text{big}}$). We call the resulting set $S''$.
Note that since $|C|$ is at most $k^2$, we have $|N_1 \setminus C| > k$. Therefore, there exists at least one vertex $u$ in $N_1 \setminus C$ such that
$N[u] \cap S = N[u] \cap S'' = \emptyset$.
Consequently, we have $D(u, S'') \subseteq C$ (at distance two) and we can apply Lemma~\ref{lem-move-far} to move one token from $C$ to $u$
and then reverse the slides of the tokens originally in $N_1 \setminus C$. We repeat this procedure as long as there are tokens in $C$.
In the second stage, we apply a similar procedure to move all tokens in $A$ to some vertex in $C$ and then from $C$ to a vertex in $N_1 \setminus C$.
This is possible because after sliding the tokens originally in $N_1 \setminus C$ to their corresponding neighbors in $N_2$ the
vertices in $A$ become closest to vertices in $C$ (at distance two).
\end{proof}

\begin{lemma}\label{lem-move-highdeg5}
Let $G$ be a connected bipartite $C_4$-free graphs and let $v \in V(G)$ be a vertex of degree at least $k^{2} + k + 1$.
Let $S$ and $T$ be two unlocked independent sets of $G$ of size $k$ such that $S \subseteq N_1$ and $T \subseteq N_1$.
Then, $S \leftrightsquigarrow T$ and this sequence can be computed in polynomial time.
\end{lemma}

\begin{proof}
As long as there exists $u \in S \setminus T$ and $w \in T \setminus S$ we can slide $u$ to $w$ as follows.
Slide all tokens (except $u$) to one of their neighbors in $N_2$. Then slide $u$ to $v$ and then slide from $v$ to $w$.
Finally, reverse all the other slides from $N_2$ to $N_1$.
\end{proof}

\begin{lemma}\label{lem-move-highdeg6}
Let $G$ be a connected bipartite $C_4$-free graphs and let $v \in V(G)$ be a vertex of degree at least $k^{2} + k + 1$.
Let $S$ be an unlocked independent set of $G$ of size $k$ such that $S \cap B(v,3) \subseteq N_1$.
Then, there exists $S'$ such that $S \leftrightsquigarrow S'$ and $S' \subseteq N_1$.
\end{lemma}

\begin{proof}
We let $X = S \cap B(v,3)$.
Since $S \cap B(v,3) \subseteq N_1$, every vertex of $S \setminus X$ is at distance at least three from $X$.
We compute the shortest path from every vertex in $S \setminus X$ to every vertex in $N_1 \setminus X$.
We let $(u,w)$ denote a pair with the minimum distance, where $u \in N_1 \setminus X$ and $w \in S \setminus X$.
If $w$ is uniquely closest to $u$ then there are two cases to consider:
\begin{itemize}
\item[1.1] When the shortest path from $u$ to $w$ does not intersect with $N[X]$ then we simply slide $u$ to $w$.
\item[1.2] Otherwise, if the shortest path $P$ intersects with $N[X]$, then there exists a
first vertex $x \in X$ such that $P \cap N_G[x] \neq \emptyset$. Therefore, we apply Lemma~\ref{lem-move-highdeg5}
in $G[B(v,r)]$ to transform $X$ into a set $X'$ such that $u \in X'$ and $x \not\in X'$. Then we can safely slide $w$ to $x$.
\end{itemize}
Now if $w$ is not uniquely closest to $u$, we assume without loss of generality that $D(u,S) \subseteq L$. Recall that
$D(u,S)$ is contained in $V(G) \setminus B(v,3)$ (at distance at least three from $u$).
We apply Lemma~\ref{lem-switch-bipartite} in $G' = (L' \cup R', E) = G - (\{v\} \cup N_1 \cup N_2)$ which guarantees
that all tokens in $L'$ will move to $R'$ via a single slide.
Hence, there must exists a first index $j$, in this sequence, where the corresponding independent set $I_j$ falls into one of the following three cases:
\begin{itemize}
\item[2.1] $|I_j \cap M_{\text{small}}| = 1$;
\item[2.2] $|I_j \cap M_{\text{big}}| = 1$; or
\item[2.3] $I_j \cap N_3 = \emptyset$ and there exists a token in $I_j$ which is uniquely closest to $u$.
\end{itemize}
In case (2.1), we apply Lemma~\ref{lem-move-highdeg4}, for case (2.2) we apply Lemma~\ref{lem-move-highdeg3}, and finally for
case (2.3) we apply either case (1.1) or case (1.2).
This completes the proof.
\end{proof}

\begin{lemma}\label{lem-move-highdeg-final}
Let $\mathcal{I} = ((L \cup R, E), S, T, k)$ be an instance of {\sc Token Sliding} where $G$ is a connected bipartite $C_4$-free graphs.
If there exists a vertex $v \in V(G)$ of degree at least $k^{2} + k + 1$ then $\mathcal{I}$ is a yes-instance.
\end{lemma}

\begin{proof}
Using Lemmas~\ref{lem-move-highdeg1} to~\ref{lem-move-highdeg5} we tranform $S$ to $S'$ and $T$ to $T'$ such that
$S' \subseteq N_1$ and $T' \subseteq N_1$. Then we transform $S'$ to $T'$ by invoking Lemma~\ref{lem-move-highdeg-final}.
\end{proof}

\begin{theorem}\label{thm-c4free-bipartite}
\sloppy%
{\sc Token Sliding} parameterized by $k$ admits a kernel with $\Oh(k^{12})$ vertices on bipartite $C_4$-free graphs.
\end{theorem}

\begin{proof}
Let $\mathcal{I} = (G, S, T, k)$ be an instance of {\sc Token Sliding} where $G$ is a bipartite $C_4$-free graphs.
We assume, without loss of generality, that $G$ is connected and $S$ and $T$ are unlocked; we can solve connected components independently and
we can return a trivial no-instance if $R(G,S) \neq R(G,T)$ (Lemma~\ref{lem-rigid-bipartite}).
Next, from Lemma~\ref{lem-move-highdeg-final}, we know that each vertex has maximum degree $\Oh(k^{2})$; otherwise we can return a trivial yes-instance.
Finally, we invoke Theorem~\ref{thm-bounded-degree-bipartite} to obtain the required kernel.
\end{proof}

%% file: paper-negative-1.tex
\section{Hardness results}\label{sec-negative}

\subsection{{\sc Token Sliding} and {\sc Token Jumping} on $C_4$-free graphs}\label{sec:neg1}

In the \textsc{Grid Tiling} problem we are given an integer $k\geq 0$ and $k^2$ sets $S_{i,j}\subseteq [m]\times[m]$, for $0\leq i,j \leq k-1$, of cardinality
$n$ called \textit{tiles} and we are asked whether it is possible to find an
element $s^*_{i,j}\in S_{i,j}$ for every $0\leq i,j \leq k-1$ such that $s^*_{i,j}$ and $s^*_{i,j+1}$ share the same first coordinate
while $s^*_{i,j}$ and $s^*_{i+1,j}$ share the same second coordinate for each $0\leq i,j \leq k-1$ (including modulo $k$).
It was proven in \cite{pcbook} that \textsc{Grid Tiling} parameterized by $k$ is $W[1]$-hard.
We prove the next theorem via a reduction from \textsc{Grid Tiling}.
Following the construction in \cite{DBLP:conf/iwpec/BonnetBCTW18} to give a graph $G$ with the desired properties and extending it to a
$\{C_4,\dots,C_p\}$-free graph $G'$ which gives a reduction to \textsc{Token Sliding}.

\begin{theorem}\label{thrm-gt-to-ts}
For any $p\geq 4$, \textsc{Token Sliding} is $W[1]$-hard on $\{C_4,\dots,C_p\}$-free graphs.
\end{theorem}

\paragraph{Construction of $G$.}
Given an instance of \textsc{Grid Tiling}, $S_{i,j}\subseteq [m]\times[m]$ $(0\leq i,j \leq k-1)$ and an integer $p\geq 4$, we use the
construction described in~\cite{DBLP:conf/iwpec/BonnetBCTW18} to create a graph $G$ with the following properties:
\begin{itemize}
\item \textbf{P1} - $G$ can be partitioned into $8k^2(p + 1)$ cliques $V_1,\dots,V_{8k^2(p + 1)}$ of size $n$ with some edges between them.
\item \textbf{P2} -  $G$ is $\{C_4,\dots,C_p\}$-free.
\item \textbf{P3} - The instance of \textsc{Grid Tiling} has a solution if and only if $\exists I\subseteq V(G)$, such that $I$ is an independent set of size $8k^2(p+1)$
\end{itemize}
Note that as each $V_i$ is a clique, any independent set of $G$ can have at most one vertex in every $V_i$.

\paragraph{Construction of $G'$.}
For $k'=8k^2(p+1)$, we construct an instance of \textsc{Token Sliding} $(G', S, T, k' + (3k'+1)\frac{p}{2} + \frac{p}{2})$ by extending the graph $G$ to a new graph $G'$.
We label the $k'$ cliques in $G$ arbitrarily as $V_1,\dots, V_{k'}$. For each $1\leq i\leq k'$ we add two vertices $x_i$ and $y_i$ adjacent to
all vertices in $V_i$. These will respectively be starting and ending positions of tokens.
Informally, we want to force all the tokens to be in their respective $V_i$ at the same time to obtain an independent set in $G$ of size $k'$.
We do this by creating \textit{guard paths}, which are paths on $p$ vertices that will be alternating between starting and target positions of tokens.
Note that we can assume $p$ is even, since if $p$ is odd we can use $p+1$ instead to create a graph which is $\{C_4,\dots,C_p\}$-free.
Let $P_G$ be a guard path with vertices $g_1,\dots, g_{p}$ and for each $x_i$ let $P_{x_i}$ be a guard path
with vertices $x_{i1},\dots, x_{ip}$ such that $x_i$ is adjacent to $x_{ip}$ and $g_p$ is adjacent to $x_{i1}$.
For each $y_i$ let $P_{y_i}$ be a guard path with vertices $y_{i1},\dots, y_{ip}$ such that $y_i$ is adjacent to $y_{i1}$ and $g_1$ is
adjacent to $y_{ip}$. Finally, for each $i$ let $P_{z_i}$ be a guard path between $x_i$ and $y_i$ with vertices $z_{i1},\dots,z_{ip}$ such
that $x_i$ is adjacent to $z_{ip}$ and $y_i$ is adjacent to $z_{i1}$. This completes the construction of $G'$. The source independent
set $S$ is the set containing all of the $x_i$ and all of the guard path vertices with odd indices:
\[
	S=\bigcup_i \{x_i, x_{ij}, y_{ij}, z_{ij}, g_j \mid j\text{ is odd}\}\,.
\]

The target independent set $T$ consists of all of the $y_i$ and all of the guard path vertices with even indices:
\[
	T=\bigcup_i \{y_i, x_{ij}, y_{ij}, z_{ij}, g_j \mid j\text{ is even}\}\,.
\]

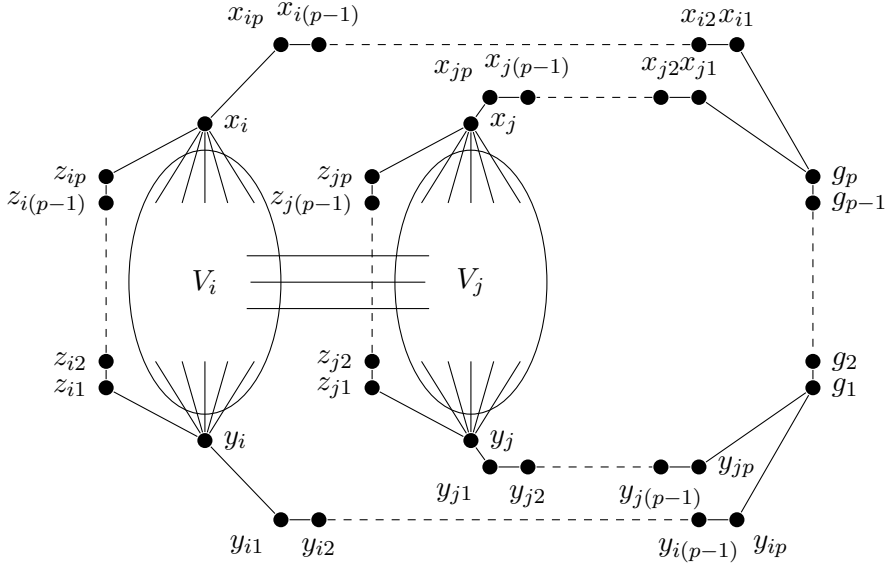
\begin{figure}
\centering
\usetikzlibrary{shapes}
\tikzstyle{clique}=[ellipse,draw,minimum height=3.5cm, minimum width=2cm]

\begin{tikzpicture}[yscale=.7,vertex/.style={circle, fill=black,thick, inner sep=2pt, minimum size=0.1cm}]
\node (Vi) at (2,2) [clique] {$V_i$};
\node (xi) at (2,5) [vertex][label=right:$x_i$]{};
\node (yi) at (2,-1) [vertex][label=right:$y_i$]{};
\node (zi1) at (0.7,0) [vertex][label=left:$z_{i1}$]{};
\node (zi2) at (0.7,0.5) [vertex][label=left:$z_{i2}$]{};
\node (zip-1) at (0.7,3.5) [vertex][label=left:$z_{i(p-1)}$]{};
\node (zip) at (0.7,4) [vertex][label=left:$z_{ip}$]{};
\draw (yi) -- (zi1) -- (zi2);
\draw [dashed] (zi2)--(zip-1);
\draw (zip-1) -- (zip) -- (xi);
\foreach \d in {1.35,1.7,2,2.3,2.65}
{
	\draw (xi)--(\d, 3.5);
	\draw (yi)--(\d, 0.5);
}

\node (Vj) at (5.5,2) [clique] {$V_j$};
\node (xj) at (5.5,5) [vertex][label=right:$x_j$]{};
\node (yj) at (5.5,-1) [vertex][label=right:$y_j$]{};
\node (zj1) at (4.2,0) [vertex][label=left:$z_{j1}$]{};
\node (zj2) at (4.2,0.5) [vertex][label=left:$z_{j2}$]{};
\node (zjp-1) at (4.2,3.5) [vertex][label=left:$z_{j(p-1)}$]{};
\node (zjp) at (4.2,4) [vertex][label=left:$z_{jp}$]{};
\draw (yj) -- (zj1) -- (zj2);
\draw [dashed] (zj2)--(zjp-1);
\draw (zjp-1) -- (zjp) -- (xj);
\foreach \d in {4.85,5.2,5.5,5.8,6.15}
{
	\draw (xj)--(\d, 3.5);
	\draw (yj)--(\d, 0.5);
}

\draw (2.6, 2) -- (4.9, 2);
\draw (2.55, 2.5) -- (4.95, 2.5);
\draw (2.55, 1.5) -- (4.95, 1.5);

\node (g1) at (10,0) [vertex][label=right:$g_1$]{};
\node (g2) at (10,0.5) [vertex][label=right:$g_2$]{};
\node (gp-1) at (10,3.5) [vertex][label=right:$g_{p-1}$]{};
\node (gp) at (10,4) [vertex][label=right:$g_{p}$]{};

\draw (gp) -- (gp-1);
\draw [dashed] (gp-1)--(g2);
\draw (g1) -- (g2);

\node (xj1) at (8.5,5.5) [vertex][label=above:$x_{j1}$]{};
\node (xj2) at (8,5.5) [vertex][label=above:$x_{j2}$]{};
\node (xjp-1) at (6.25,5.5) [vertex][label=above:$x_{j(p-1)}$]{};
\node (xjp) at (5.75,5.5) [vertex][label=above left:$x_{jp}$]{};

\node (xi1) at (9,6.5) [vertex][label=above:$x_{i1}$]{};
\node (xi2) at (8.5,6.5) [vertex][label=above:$x_{i2}$]{};
\node (xip-1) at (3.5,6.5) [vertex][label=above:$x_{i(p-1)}$]{};
\node (xip) at (3,6.5) [vertex][label=above left:$x_{ip}$]{};

\draw (xi) -- (xip) -- (xip-1);
\draw [dashed] (xi2)--(xip-1);
\draw (gp) -- (xi1) -- (xi2);

\draw (xj) -- (xjp) -- (xjp-1);
\draw [dashed] (xj2)--(xjp-1);
\draw (gp) -- (xj1) -- (xj2);

\node (yjp) at (8.5,-1.5) [vertex][label=right:$y_{jp}$]{};
\node (yjp-1) at (8,-1.5) [vertex][label=below:$y_{j(p-1)}$]{};
\node (yj2) at (6.25,-1.5) [vertex][label=below:$y_{j2}$]{};
\node (yj1) at (5.75,-1.5) [vertex][label=below left:$y_{j1}$]{};

\node (yip) at (9,-2.5) [vertex][label=below right:$y_{ip}$]{};
\node (yip-1) at (8.5,-2.5) [vertex][label=below:$y_{i(p-1)}$]{};
\node (yi2) at (3.5,-2.5) [vertex][label=below:$y_{i2}$]{};
\node (yi1) at (3,-2.5) [vertex][label=below left:$y_{i1}$]{};

\draw (yi) -- (yi1) -- (yi2);
\draw [dashed] (yi2)--(yip-1);
\draw (g1) -- (yip) -- (yip-1);

\draw (yj) -- (yj1) -- (yj2);
\draw [dashed] (yj2)--(yjp-1);
\draw (g1) -- (yjp) -- (yjp-1);

\end{tikzpicture}
\caption{The construction of $G'$ for two cliques $V_i,V_j$ in $G$.} \label{fig:GPrime}
\end{figure}

\begin{lemma}\label{lem:G'Cl-free}
For any $p\geq 4$, $G'$ is $\{C_4,\dots,C_p\}$-free.
\end{lemma}

\begin{proof}
By \textbf{P2}, $G$ is $\{C_4,\dots,C_p\}$-free.
Any cycle which contains a vertex on one of the guard paths has length greater than $p$.
Thus if a cycle of length $\ell$ exists for some $4\leq \ell\leq p$ it must only have vertices
in $V(G)\cup\{x_i,y_i \mid 1\leq i \leq k'\}$ and contains at least one of $x_i$ or $y_i$.
Assume, without loss of generality, that it contains $x_i$, then the vertices adjacent to $x_i$ in the cycle must be
in $V_i$. As $V_i$ is a clique the cycle contains a $C_3$ so is not induced.
\end{proof}

\begin{lemma}\label{lem:tiling-to-reconfig}
If there is a solution to the \textsc{Grid Tiling} instance then there is a reconfiguration sequence from $S$ to $T$ in $G'$.
\end{lemma}

\begin{proof}
By \textbf{P3}, there exists an independent set $I$ containing one vertex $v_i$ in every $V_i$. This gives the following reconfiguration sequence from $S$ to $T$.
\begin{enumerate}
\item Move each token on $x_i$ to $v_i$.
\item Move the tokens along the guard paths: for all odd $j$ starting with the greatest $j$ values move the token on each $z_{ij}$ to $z_{i(j+1)}$, then move the tokens on $x_{ij}$ to $x_{i(j+1)}$, $g_{j}$ to $g_{j+1}$ and finally $y_{ij}$ to $y_{i(j+1)}$
\item Move each token on $v_i$ to $y_i$.
\end{enumerate}
This completes the proof.
\end{proof}

Finally let us prove the converse direction.
For each $i$, let $W_i := \{x_i, y_i\} \cup V_i$. Let us first show that in any valid reconfiguration sequence
the tokens initially on the guard paths, $P_{x_i}$, $P_{y_i}$, $P_{z_i}$, and $P_G$ are stuck on their respective paths.
We first need the following simple observation.

\begin{observation}\label{obs:nb-tok-path}
Let $I$ be an independent set of $G'$ of size $k' + (3k'+1)\frac{p}{2}$ such that for every $i \leq k'$, $|W_i \cap I| = 1$. Then for every guard path $P$ of $G'$ we have $|I \cap P| = \frac{p}{2}$.
\end{observation}

\begin{proof}
We assume there are exactly $k'$ tokens on $\cup_{i=1}^{k'} W_i$. Then since for any guard path $P$ we have $|P \cap I| \leq \frac{p}{2}$, there must be exactly $\frac{p}{2}$ tokens on each of the $3k'+1$ guard paths.
\end{proof}

\begin{lemma}\label{lem:stuck-path}
Let $I_1, I_2, \ldots, I_e$ be a valid reconfiguration sequence such that $I_1 = S$ and $I_e = T$. For every $s \leq e$ and for every $i \leq k'$, $|W_i \cap I_s| = 1$.
\end{lemma}
\begin{proof}
By construction the statement is true for $I_1$. Consider the smallest integer $s \leq e$ such that $I_e$ does not satisfy the condition of Lemma \ref{lem:stuck-path}.
By this choice of $s$ we have $|W_j \cap I_r| = 1$ for every $r < s$ and every $j \leq k'$, hence there exists a
unique $i \leq k'$ such that $|W_i \cap I_s| = 0$ or $|W_i \cap I_s| = 2$. Let us show that we obtain a contradiction in both cases:\newline
Case $1$: $|W_i \cap I_s| = 0$. Since $|W_j \cap I_s| = 1$ for every $j \neq i$, there can be no move from $V_i$ to $V_j$ if there is a token
on $V_i$ in $I_{s-1}$. So there must be a token on one of $x_i$, $y_i$ in $I_{s-1}$ and this token must move on an adjacent guard path $P$.
But then since $I_{s-1}$ satisfies the condition of Observation \ref{obs:nb-tok-path}, we have $|P \cap I_{s}| = \frac{p}{2} + 1$, a contradiction.\newline
Case $2$: $|W_i \cap I_s| = 2$. If $|V_i \cap I_{s-1}| = 1$ then by construction no token can
move to $V_i$ between times $s-1$ and $s$. Hence we can suppose w.l.o.g that $I_{s-1} \cap W_i = \{x_i\}$ and $I_s \cap W_i = \{x_i, y_i\}$.
So it must be that a token moves either from $y_{i1}$ to $y_i$ or from $z_{i1}$ to $y_i$ at time $s-1$ and then either $z_{i1} \notin I_{s-1}$ or $y_{i1} \notin I_{s-1}$.
In both case, since $I_{s-1}$ satisfies the condition of Observation \ref{obs:nb-tok-path}, we obtain that there must be a token on every vertex with even index
on the guard paths $P_{x_i}$, $P_{y_i}$, $P_{z,i}$ and $P_G$. In particular we have $\{x_{ip}, x_i\} \subseteq I_s$, a contradiction.
\end{proof}

\begin{lemma}\label{lem:reconfig-to-tiling}
If there is a reconfiguration sequence from $S$ to $T$ in $G'$ then there is a solution to the \textsc{Grid Tiling} instance.
\end{lemma}

\begin{proof}
Given the reconfiguration sequence $I_1, I_2, \ldots, I_e$ such that $I_1 = S$ and $I_e = T$ let
us consider the last time $t-1$ at which a token moves from $x_i$ for some $i \leq k'$.
Such a time exists since all the tokens must move at least one time in a reconfiguration sequence from $S$ to $T$.
By Lemma \ref{lem:stuck-path} this token moves from $x_i$ to $V_i$. In particular, there is no token on $x_{ip}$ in $I_{t-1}$, and since $I_{t-1}$ satisfies
the condition of Observation \ref{obs:nb-tok-path}, there must be a token on $g_{s}$ for every $s$ odd.
This in turn implies that for every $j$ there must be a token on $y_{js}$ for every $s$ odd and in particular for $y_{j1}$, so there cannot be a token
on any $y_j$. Thus for every $j \leq k'$, $|V_j \cap I_t| = 1$ by Lemma \ref{lem:stuck-path}, giving an independent set of size $k'$ in $G$.
By \textbf{P3}, we know that this implies a solution for the \textsc{Grid Tiling} instance.
\end{proof}

The combination of Lemmas \ref{lem:G'Cl-free}, \ref{lem:tiling-to-reconfig} and \ref{lem:reconfig-to-tiling} give us the result of Theorem \ref{thrm-gt-to-ts}.

\begin{lemma}\label{lem:maximum-IS}
Let $I$ be an independent set of $G'$ of size $k' + (3k'+1)\frac{p}{2}$ then $I$ is a maximum independent set of $G'$.
\end{lemma}

\begin{proof}
First note that, by Observation \ref{obs:nb-tok-path}, $I$ has $\frac{p}{2}$ tokens on every guard path
and exactly one token in every $W_i := \{x_i, y_i\} \cup V_i$.
Assume $I$ is not maximum, so there is some independent set $I'$ of $G'$ with $|I'|>|I|$.
The maximum size of an independent set on a path of length $p$ is $\frac{p}{2}$, so $I'$ must have 2 tokens in some $W_i$ which must be on $x_i$ and $y_i$.
However this implies that there can only be $\frac{p}{2} - 1$ tokens in $I'$ on $P_{z_i}$. Thus $|I'|\leq |I|$.
\end{proof}

\begin{corollary}
\sloppy%
For any $p\geq 4$, \textsc{Token Jumping} is $W[1]$-hard on $\{C_4,\dots,C_p\}$-free graphs.
\end{corollary}

\begin{proof}
$G'$ is a single fully-connected component and by Lemma \ref{lem:maximum-IS} the starting set $S$ is a maximum set of $G'$.
Thus the \textsc{Token Sliding} instance is equivalent to a \textsc{Token Jumping} instance and the reduction from \textsc{Grid Tiling} holds.
\end{proof}

%% file: paper-negative-2.tex
\subsection{{\sc Token Sliding} on bipartite graphs}\label{sec:neg2}
This section is devoted to proving the following theorem:

\begin{theorem}\label{thm:isr-ts-bipartite}
\textsc{Token sliding} on bipartite graphs is W[1]-hard parameterized by $k$.
\end{theorem}

The proof of Theorem \ref{thm:isr-ts-bipartite} consists in a reduction from \textsc{Multicolored Independent Set}.
In what follows, $\mathcal{I} := (G, k , (V_1, \ldots, V_k))$ denotes an instance of \textsc{Multicolored Independent Set}, which is known to be $W[1]$-hard
parameterized by $k$ \cite{pcbook}. In Section \ref{sec:construction}, we detail the construction of the equivalent
instance $\mathcal{I'} := (G', I_s, I_e, 4k+2)$ of \textsc{Token Sliding}, where $G'$ is a bipartite graph and $I_s$, $I_e$ are independent
sets of size $4k+2$, and we prove that if $\mathcal{I}$ is a yes-instance then $\mathcal{I'}$ is a yes-instance.
The more involved proof of the converse direction is detailed in Sections \ref{sec:well-organized-conf} and \ref{sec: bad-moves}.

\subsubsection{Construction of G'}\label{sec:construction}
In what follows, $V(G') := (\mathcal{A}, \mathcal{B})$ denotes the bipartition of $G'$.
For every $p \in \{1, \ldots, k\}$, both $\mathcal{A}$ and $\mathcal{B}$ contain two copies of
the set $V_p$ denoted as $A_{2p-1}$, $A_{2p}$ and $B_{2p-1}, B_{2p}$ respectively, plus some additional
vertices that will be described in the next subsection. Two vertices $u',v' \in V(G')$ are said to be \emph{equivalent}
and we write $u' \sim v'$ if and only if they are copies of the same vertex in $G$. With this definition, every
vertex $u \in V_p$ has exactly four copies in $G'$ (one in each copy of $V_p$). Note that the $\sim$ relation is
transitive and symmetric. We also define the sets $A := \cup_{p=1}^k A_{2p-1} \cup A_{2p}$ and $B := \cup_{p=1}^k B_{2p-1} \cup B_{2p}$.
For every vertex $u'$ of $A \cup B$, the \emph{corresponding vertex} of $u'$ denoted as $orr(u')$ is the unique vertex $u \in V(G)$ that $u'$ is a copy of.
With these definitions at hand, we can now explain how the copies of the sets $V_1, V_2, \ldots, V_k$ are connected in $G'$.
For every two vertices $u' \in A_i$ and $v' \in B_j$ there is an edge connecting $u'$ to $v'$ in $G'$ if and only if:

\begin{enumerate}
    \item $A_i$ and $B_j$ are not copies of the same subset of $V(G)$ and $(orr(u'), orr(v')) \in E(G)$, or
    \item $A_i$ and $B_j$ are copies of the same subset of $V(G)$ and $u' \nsim v'$.
\end{enumerate}

In other words, if $A_i$ and $B_j$ are not copies of the same subset, we connect these sets in the same way there corresponding sets are connected in $G$.
If at the contrary $A_i$ and $B_j$ are copies of the same subset, then $G'[A_i \cup B_j]$ induces a complete bipartite graph minus the matching consisting
of every two pairs of equivalent vertices in $A_i \cup B_j$. The connection between four copies of the same subset of $V(G)$ is illustrated in
Figure \ref{fig:connections-copies}. Let us explain how we make use of such a construction. The following observation follows directly from the definition of $G'$:

\begin{observation}\label{obs:construction-multicolored}
Let $I'$ be an independent set of $G'$ such that for every $p \in 1, 2, \ldots, k$ we have $I' \cap A_{2p-1} = \{u_{2p-1}\}$
and $I' \cap B_{2p-1} = \{v_{2p-1}\}$. Then the set $I := \{orr(u_1), \ldots, orr(u_k)\}$ is a multicolored independent set of $G$.
\end{observation}

\begin{proof}
For any two $i,j \in 1,2, \ldots, k$, $u_{2i-1}$ and $v_{2j-1}$ are non-neighbor in $G'$ since $I'$ is an independent set.
Furthermore, if $i \neq j$ then $A_{2i-1}$ and $B_{2j-1}$ are not copies of the same subset of $V(G)$ and
thus $orr(u_{2i-1}) \neq orr(v_{2j-1})$, so the set $I$ contains $k$ distinct vertices of $G$.
Since $orr(u_{2j-1}) = orr(v_{2j-1})$, we have that $(orr(u_{2i-1}), orr(u_{2j-1})) \notin E(G)$ for any two $i \neq j$, and
since $orr(u_{2i-1}) \in V_{2i-1}$ by construction, the set $I$ is a multicolored independent set of $G$.
\end{proof}

Observation \ref{obs:construction-multicolored} ensures that any independent set of a reconfiguration
sequence of $G'$ having exactly one vertex in $A_{2p-1}$ and one vertex in $B_{2p-1}$ for every $p \in 1,2, \ldots k$ corresponds
to a multicolored independent set of $G$. Note that up to that point, we did not make use of the sets $A_{2p}$ and $B_{2p}$.
The following observation explains why we need two copies of every $V_p$ in both sides of the bipartition:

\begin{observation}\label{obs:construction-equivalent-lock}
Let $I'$ be an independent set of $G'$ and $p \in 1,2, \ldots, k$ such that $I' \cap A_{2p-1} = \{u_{2p-1}\}$, $I' \cap A_{2p} = \{u_{2p}\}$, and $u_{2p-1} \sim u_{2p}$.
Then the tokens on $u_{2p-1}$ and $u_{2p}$ cannot move to $B$.
\end{observation}

\begin{proof}
By construction $N(u_{2p-1}) \cap B = N(u_{2p}) \cap B$ since these two vertices are equivalent.
It follows that none of the two tokens on $u_{2p}$ nor $u_{2p-1}$ can move to $B$.
\end{proof}

If at some point in the reconfiguration sequence two tokens are positioned on equivalent vertices in $A$, then these tokens lock each other at their respective position in some sense. Note that by symmetry of the construction, the same observation can be made when two tokens are positioned on equivalent vertices in $B$. On the contrary, if two tokens on the same copies of $V_p$ in $A$ are positioned on two non-equivalent vertices we have the following:
\begin{observation}\label{obs:construction-noneq-lock}
Let $I'$ be an independent set of $G'$ and $p \in 1, 2, \ldots k$ such that $I' \cap A_{2p-1} = \{u_{2p-1}\}$, $I' \cap A_{2p} = \{u_{2p}\}$, and $u_{2p-1} \nsim u_{2p}$. Then $I' \cap (B_{2p-1} \cup B_{2p}) = \emptyset$.
\end{observation}

\begin{proof}
By construction $B_{2p-1} \cup B_{2p} \subseteq N(u_{2p-1}) \cup N(u_{2p})$ since these two vertices are not equivalent.
\end{proof}

This observation not only ensures that $B_{2p-1} \cup B_{2p} = \emptyset$ but also ensures that no other token but the ones positioned
on $u_{2p-1}$ and $u_{2p}$ can move to $B_{2p-1} \cup B_{2p}$. Then, by Observations \ref{obs:construction-equivalent-lock} and \ref{obs:construction-noneq-lock}, either
there are two tokens on equivalent vertices in $A_{2p-1} \cup A_{2p}$ and then these tokens cannot move to $B$ (and ensures that if there is a token
on $B_{2p-1} \cup B_{2p}$ it must be on an equivalent vertex), or there are two tokens on non-equivalent vertices forbidding any other token to move to $B_{2p-1} \cup B_{2p}$.

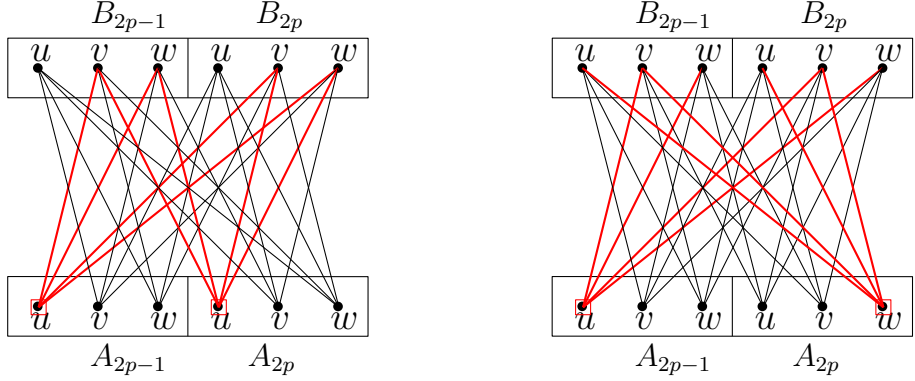
\begin{figure}[tbph]
    \centering
    \begin{minipage}[b]{.45\textwidth}
        \centering
        \input{connexion_equiv.tex}
    \end{minipage}%
    \begin{minipage}[b]{0.45\textwidth}
        \centering
        \input{connexion_not_equiv.tex}
    \end{minipage}

    \caption{Connections between the four copies of $V_p$ in $A \cup B$. Vertices with the same name are equivalent vertices. The red square represent tokens: two tokens are positioned on equivalent vertices at the left, and on non-equivalent vertices at the right.}
    \label{fig:connections-copies}
\end{figure}

\paragraph{Definition of the initial and target independent sets.}

The initial independent set $I_s$ consists in two sets of $2k$ vertices $A_{start}$ and $A_{end}$ plus two vertices $s_A, e_A$ included in $\mathcal{A}$, and the target independent set $I_e$ consists in two sets of $2k$ vertices $B_{start}$ and $B_{end}$ plus two vertices $s_B, e_B$ included in $\mathcal{B}$. The two sets $I_s$ and $I_e$ are disjoint from $A \cup B$. The graph induced by $A_{start} \cup B_{end} \cup \{s_A, e_B\}$ and the graph induced by $A_{end} \cup B_{start} \cup \{s_B, e_A\}$ are complete bipartite graphs. The main goal of this section is to explain how to connect the  set $A_{start} \cup  B_{start}$ and the set $A_{end} \cup B_{end}$ to $A \cup B$ in order to ensure that any reconfiguration sequence transforming one into the other enforces the $2k$ tokens starting on $A_{start}$ and the $2k$ tokens starting on $B_{start}$ to switch sides by going through $A \cup B$. More particularly, we will show the existence of an independent set that satisfies the condition of Observation \ref{obs:construction-multicolored} in any such reconfiguration sequence, giving a multicolored independent set of $G$. %
For $p \in 1, 2, \ldots, 2k$, we denote by $a_{s,p}$ and $b_{s,p}$ the vertices of $A_{start}$ and $B_{start}$ respectively and we denote by $a_{e,p}$ and $b_{e,p}$ the vertices of $A_{end}$ and $B_{end}$ respectively. These vertices are connected to $A \cup B$ as follows:
\begin{enumerate}
    \item the vertices $a_{s,p}$ and $a_{e,p}$ are complete to $B - \cup_{i = 1}^{p-1} B_i$, and
    \item the vertices $b_{s,p}$ and $b_{e,p}$ are complete to $A - \cup_{i = 1}^{p-1} A_i$.
\end{enumerate}
An illustration of the full construction is given in Figure \ref{fig:complete_schem}. By construction, no token starting on $A_{start} \cup \{s_A\}$ can move to $B_{end} \cup \{e_B\}$ as long as there are at least two tokens on $A_{start} \cup \{s_A\}$ (and the same goes for $B_{start} \cup \{s_B\}$ and $A_{end} \cup \{e_A\}$). Since there are initially $2k + 1$ tokens on $A_{start} \cup \{s_A\}$ and since $N(s_A) \cap B = \emptyset$, the $2k$ tokens initially on $A_{start}$ must move to $B$ at some point in the sequence, and the same goes for $B_{start}$ and $A$. The tokens initially on $s_A$ and $s_B$ have a special role and act as "locks": without these token, the last token remaining on $A_{start}$ (resp. $B_{start}$) would be able to move directly to $B_{end}$ without never going through $B$ (resp. $A$). Let us now explain the connections to $A \cup B$.
\begin{observation}\label{obs:construction-astart-frozen}
\sloppy%
Let $I'$ be an independent set of $G'$ such that $\{a_{s,p}, a_{s,p+1}, \ldots, a_{s, 2k}\} \subseteq I'$ for some $p < 2k$. Then the tokens on $\{a_{s,p+1}, a_{s,p+2}, \ldots, a_{s, 2k}\}$ are frozen. Furthermore the token on $a_{s,p}$ cannot move to $\cup_{i=p+1}^{2k} B_p$.
\end{observation}
\begin{proof}
Let $q > p$ and suppose there is a token on $a_{s,q}$. This token cannot move to $B_{end}$ nor $e_B$ since there is a token on $a_{s,p}$ with $p < q$ and $G'[A_{start} \cup B_{end} \cup \{s_A, e_B\}]$ induces a complete bipartite graph. By construction $N(a_{s,q}) \subseteq N(a_{s,p})$ hence the token on $a_{s,q}$ cannot move to $B$ and this token is frozen. The second statement follows from the fact that $\cup_{i=p+1}^{2k} B_p \subseteq N(a_{s,p}) \cap N(a_{s,p+1})$.
\end{proof}
\begin{figure}
    \centering
    \includegraphics[scale=0.7]{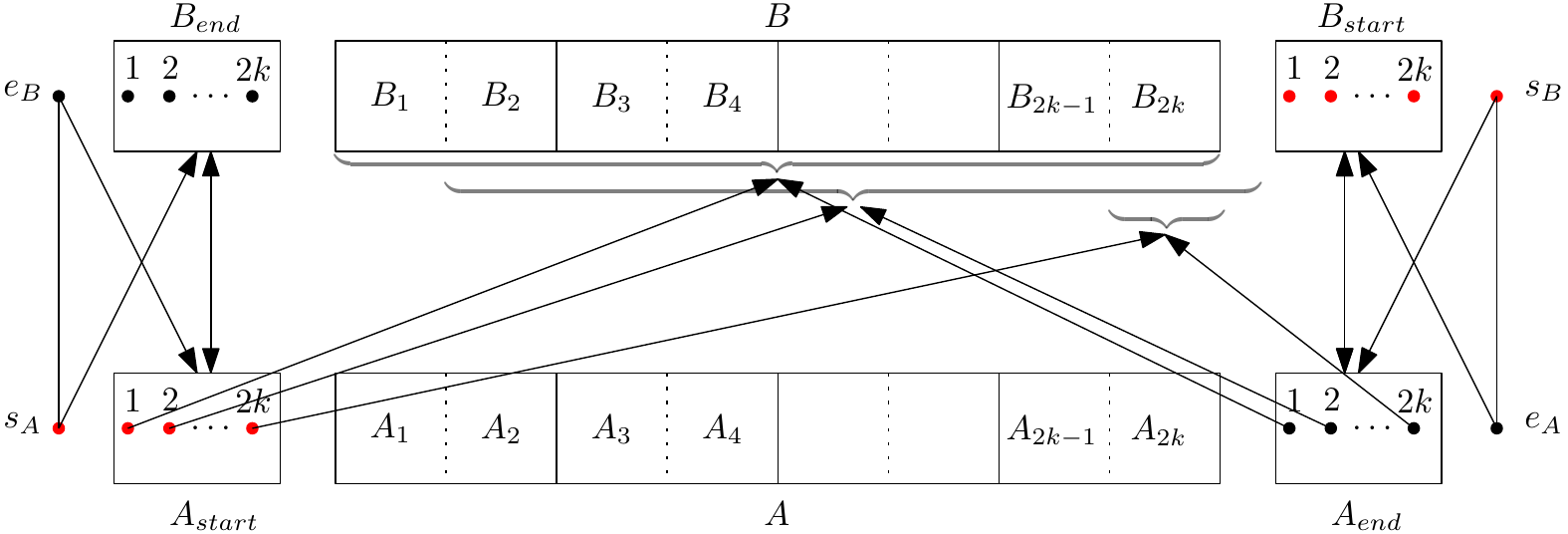}
    \caption{The constructed graph $G'$. Vertices in red are the vertices of $I_s$. An arrow between a vertex $v$ and a subset of vertices indicates that $v$ is complete to this subset. An arrow between a vertex $v$ and a brace indicates that $v$ is complete to the subsets included in the brace. A double arrow between two sets indicate these sets induce a complete bipartite graph. The connections between $A$ and $B_{end} \cup B_{start}$ are symmetric and have been omitted for the sake of clarity.}
    \label{fig:complete_schem}
\end{figure}

By symmetry, the same observation can be made for tokens on $B_{start}$.
This shows that the tokens initially on $A_{start}$ and $B_{start}$ must respect a strict order to move respectively to $B$ and $A$: the only tokens that can
initially move are the tokens on $a_{s,1}$ and $b_{s,1}$ and these have no choice but to move to $B_1$ and $A_1$ respectively.
After such a move the tokens on $a_{s,2}$ and $b_{s,2}$ are free to move to $B_2$ and $A_2$ respectively, and so on.
Suppose that after the first $4$ moves, there is exactly one token in each of the four subset $A_1$, $B_1$, $A_2$ and $B_2$.
Then it is not hard to see - but will be formally proved in the next section - that these tokens lie on equivalent vertices, corresponding to a
unique vertex of $G$. By Observation \ref{obs:construction-equivalent-lock} these tokens cannot move to the other side of the bipartite graph and
must stay at the same position while the remaining tokens on $A_{start}$ and $B_{start}$ moves to $A \cup B$. With the full
constructions of $G'$, $I_s$ and $I_e$ at hand we can prove the direct part of the reduction:

\begin{lemma}
If there is a multicolored independent set of size $k$ in $G$ then there exists a reconfiguration sequence transforming $I_s$ to $I_e$ in $G'$.
\end{lemma}
\begin{proof}
Let $u_1 \in V_1, \ldots, u_k \in V_k$ be a multicolored independent set of $G$. For $p$ in $1 \ldots, k$, let $u'_{2p-1}, u'_{2p}$ (resp. $v'_{2p-1}, v'_{2p}$) be the copies of $u_p$ in $A$ (resp. B). Consider the following sequence:
\begin{enumerate}
    \item For $p \in \{1, \ldots, k\}$ in increasing order, move the token on $a_{s,2p-1}$ to $u'_{2p-1}$, then move the token on $a_{s,2p}$ to $u'_{2p}$. Move the token on $b_{s,2p-1}$ to $v'_{2p-1}$, then move the token on $b_{s,2p}$ to $v'_{2p}$.
    \item Move the token on $s_A$ to $e_B$ then move the token on $s_B$ to $e_A$.
    \item For $p \in k, \ldots, 1$ in decreasing order, move the token on $u'_{2p}$ to $a_{e,2p}$, then move the token on  $u'_{2p-1}$ to $a_{e,2p-1}$. Move the token on $v'_{2p}$ to  $b_{e,2p}$, then move the token on $v'_{2p-1}$ to $b_{e,2p-1}$.
\end{enumerate}
\end{proof}

The remainder of the section is dedicated to the converse part of the reduction.
More particularly, we formally show that there is an independent set satisfying the condition of Observation \ref{obs:construction-multicolored} in any shortest
reconfiguration sequence transforming $I_s$ to $I_e$.

\subsubsection{Well-organized configurations}\label{sec:well-organized-conf}
To simplify the tracking of tokens along the transformation, we give different colors to the tokens initially on $A_{start}$ and $B_{start}$.
The tokens initially on $A_{start}$ are the \emph{blue tokens} and the tokens initially on $B_{start}$ are the \emph{red tokens}.
We say a vertex $v$ is \emph{dominated} by a vertex $u$ in $G$ if $v \in N_G(u)$. Similarly, we say a set $U$ is dominated by $W$ if $U \subseteq N_G(W)$.
Given a configuration $C$, $M_A(C)$ (resp. $M_B(C)$) is the maximum integer $p \in \llbracket 1, 2k \rrbracket$ such that there is a token on $A_p$ (resp. $B_p$).
By convention, if there is no token on $X \in \{A,B\}$, we set $M_X(C) = 0$. A configuration $C$ is  \emph{well-organized} if there is a token
on either $s_A$ or $e_B$ and on either $s_B$ or $e_A$ and if it satisfies the following conditions:
\begin{enumerate}
    \item For every $p \leq M_A(C)$ and every $q \leq M_B(C)$ there is exactly one token on $A_p$ and  exactly one token on $B_q$.
    \item If $M_A(C) < 2k$ then for every $M_A(C) < p \leq 2k$ there is a token on $a_{s,p}$. If $M_B(C) < 2k$ then for every $M_B(C) < q \leq 2k$ there is a token on $b_{s,q}$.
\end{enumerate}

Since both the construction and the definition of well-organized configurations are symmetric, we can always assume that $M_A(C) \leq M_B(C)$ for any well-organized configuration $C$.
Note that the initial configuration is well-organized.
We say that two configurations $C$ and $C'$ are \emph{adjacent} if $C$ can be transformed into $C'$ by moving exactly one token.\\
Throughout the proof let $S := C_1, \ldots, C_N$ denote a shortest reconfiguration sequence from $I_s$ to $I_e$. We say that a token  \emph{moves from a set $X$ to a set $Y$ at time $t$} and we write $(t: X \rightarrow Y)$ if there exists two set $X,Y \subseteq V(G')$ and two vertices $x \in X$, $y \in Y$ such that $C_{t+1} =  C_{t} - \{x\} + \{y\}$. When the sets $X$ and $Y$ contain exactly one vertex we write $(t: x \rightarrow y)$ by abuse of notation. A move that transforms a well-organized configuration into a configuration that is not well-organized is a \emph{bad move}. We aim to show the following:
\begin{lemma}\label{lem:no-bad-move-shortest}
A shortest reconfiguration sequence from $I_s$ to $I_e$ contains no bad move.
\end{lemma}

With Lemma \ref{lem:no-bad-move-shortest} at hand, the proof of the converse part of the reduction easily follows:
\begin{lemma}\label{lem:converse-direction}
If there exists a reconfiguration sequence from $I_s$ to $I_e$ in $G'$, then there exists a multicolored independent set in $G$.
\end{lemma}
\begin{proof}
Consider a shortest reconfiguration sequence $S$ from $I_s$ to $I_e$, which exists by supposition. By Lemma \ref{lem:no-bad-move-shortest} this sequence contains no bad moves, therefore all the configurations of $S$ are well-organized since the initial configuration is. Consider the configuration $C$ just before the first token reaches $A_{end} \cup B_{end}$ (which exists since $A_{end} \cup B_{end} \subseteq I_e$). By definition of well-organized configurations there can be no token on $A_{start} \cup B_{start}$ in $C$ and thus we have $M_A(C) = M_B(C) = 2k$. Then by Observation \ref{obs:construction-multicolored} there exists a multicolored independent set in $G$.
\end{proof}

The remainder of this section is dedicated to the proof of Lemma \ref{lem:no-bad-move-shortest}. Let us begin with a few observations about well-organized configurations, which will be useful throughout all the subsections:

\begin{observation}\label{obs:wo-same-vertex}
  \sloppy%
Let $C$ be a well-organized configuration. For every $p \leq M_A(C)$ we have $|A_p \cap C| = |B_p \cap C| = 1$, and the unique vertex of $A_p \cap C$ and the unique vertex of $B_p \cap C$ are equivalent.
\end{observation}
\begin{proof}
By definition of well-organized configuration, there is exactly one token on $A_p$ and one token on $B_p$ for $p \leq M_A(C)$. Let $u$ be the unique vertex of $A_p \cap C$: by construction the only vertex $v$ of $B_p$ that is not in $N(u)$ is the copy of $u$ in $B_p$. %
\end{proof}

\begin{observation}\label{obs:wo-same-vertex-bis}
  \sloppy%
Let $C$ be a well-organized configuration and $p \leq 2k$ be an odd integer such that $|A_p \cap C| = |A_{p+1} \cap C| = 1$. Then $|B_p \cap C| = |B_{p+1} \cap C| = 1$ and the four vertices in these sets are equivalent.
\end{observation}
\begin{proof}
Since $C$ is well-organized, $M_B(C) \geq M_A(C)$ and there is one token on $B_p$ and one token on $B_{p+1}$. Let $u$ (resp. $u'$) be the unique vertex of $A_{p} \cap C$ (resp. $B_{p} \cap C$) and $v$ (resp. $v'$) be the unique vertex of $A_{p+1} \cap C$ (resp. $B_{p+1} \cap C$). By Observation \ref{obs:wo-same-vertex} we have $u \sim u'$ and $v \sim v'$. By construction the only vertex of $B_{p+1}$ that is not in $N(u)$ is a copy of $u$ since $p$ is odd ($B_p$ and $B_{p+1}$ are copies of the same subset of $V(G)$). We obtain that $u \sim v'$, and the proof follows by the transitivity of the $\sim$ relation.
\end{proof}

\begin{observation}\label{obs:well-organized-frozen}
  \sloppy%
Let $C$ be a well-organized configuration. For every $p < M_A(C)$ and every $q < M_B(C)$, the token on $A_p$ and the token on $B_q$ are frozen.
\end{observation}
\begin{proof}
Let $p < M_A(C)$ and let $\{v^A_{p}\} := A_p \cap C$. Since $p < M_A(C)$, there is a token on another vertex $v^A_{p'} \in A_p'$ such that $v^A_{p'} \sim v^A_{p}$ by Observation \ref{obs:wo-same-vertex-bis}. Since these two vertices share the same neighborhood in $B$, the token on $v^A_p$ cannot move to $B$. Furthermore, there is a token on $A_{q}$ for any $q \leq p$ thus this token cannot move to $b_{s,q}$ nor $b_{e,q}$ and since $p < M_A(C)$, there is a token on $A_{p+1}$ and the token cannot go to $b_{s,p}$ nor $b_{e,p}$. It follows that it cannot move to $B_{start}$ nor $B_{end}$ and that the token on $v^A_p$ is frozen. By symmetry, the same goes for the token on $B_{p}$ for $p \leq M_A(C)$. We then have to be careful about the tokens on $B_q$ for $M_A(C) < q < M_B(C)$. Let $\{v^B_q\} := B_q \cap C$. Since $A_q \cap C = \emptyset$ for any such $q$, we cannot guarantee that $v^B_{q} \sim v^B_{q+1}$ even when $A_q$ and $A_{q+1}$ are copies of the same set. However, for any $p \leq M_A(C)$ the set $A_p - v^A_p$ is dominated by $v^B_p$ and $v^A_p \notin N(v^B_q)$ for any $q$ since $C$ is an independent set, hence the token on $B_q$ cannot move to $A_p$. Furthermore, since $b_{s, M_A(C)+1} \in C$, no token can move from $B$ to $A_p$ for any $p > M_A(C)$. It follows that the tokens on $B_q$ for $M_A(C) < q < M_B(C)$ are also frozen.
\end{proof}

\begin{observation}\label{obs:min-move-wo}
Let $C$ be any well-organized configuration reachable from $C_1$. Each token moves at most one time in a shortest reconfiguration sequence from $C_1$ to $C$.
\end{observation}
\begin{proof}
We can reach $C$ by moving the tokens in the following order: for $p \in 1, \ldots, M_B(C)$ the token on $a_{s,p}$ moves to $B_p \cap C$ and for $p \in 1, \ldots, M_A(C)$ the token on $b_{s,p}$ moves to $A_p \cap C$. Then, if $C \cap \{s_A, e_B\} = \{e_B\}$ (resp. $C \cap \{s_B, e_A\} = \{e_A\}$), move the token from $s_A$ to $e_B$ (resp. from $s_B$ to $e_A$). This is a shortest sequence since it contains exactly $|C \backslash C_1|$ moves, and every token moves at most one time.
\end{proof}

The strategy to prove Lemma \ref{lem:no-bad-move-shortest} is as follows: we show that if there is a bad move at time $t$, then there exists a time $t' > t$ at which this bad move is canceled in the sense that the configuration obtained a time $t' + 1$ is, again, well-organized. Such a reconfiguration sequence contains at least $M_A(C_{t'+1}) + M_B(C_{t'+1}) + 1$ moves since at least one token moved twice, and then Observation \ref{obs:min-move-wo} ensures that it is not a shortest sequence, contradicting our choice of $S$.
The remainder of the proof is organized as follows. In Section \ref{sec: bad-moves} we identify, up to symmetry, three different types of bad moves and give some observations about the structure of configurations obtained after such moves. In Section \ref{sec:bad-move-t1} we then
show how to cancel (in the sense mentioned above) bad moves of type $1$, and we deal with types $2$ and $3$ in Section \ref{sec:bad-move-t2-t3}.

\subsubsection{Bad moves}\label{sec: bad-moves}

\begin{observation}\label{obs:well-nonwell}
Let $t \in 1, 2, \ldots, N$ be such that the configuration $C_t$ is well-organized and $C_{t+1}$ is not. Then one of the following holds:
\begin{enumerate}
    \item $(t: A \rightarrow B)$ or $(t: B \rightarrow A)$, or
    \item $(t: A \rightarrow B_{end})$ and $B_{start} \cap C_t \neq \emptyset$ or $(t: B \rightarrow A_{end})$ and $A_{start} \cap C_t \neq \emptyset$, or
    \item $(t: s_A \rightarrow B_{start})$ or $(t: s_B \rightarrow A_{start})$.
\end{enumerate}
\end{observation}
\begin{proof}
First, there can be no move from $A_{start}$ to $B_{end}$ at time $t$. Indeed if there is a token on $A_{start}$ then there must be a token on $s_A$ since $C_t$ is well-organized  and both of these tokens dominate all of $B_{end}$. By symmetry, the same goes for $A_{end}$ and $B_{start}$. By construction, the only token that can move from $A_{start}$ is the token on $a_{s, M_B(C)+1}$ which can only go to $B_{M_{B(C)+1}}$ and such a move leads to a well-organized configuration and cannot be a bad move. Conversely, the only token that can move from $B$ is the token on $B_{M_B(C)}$ by Observation \ref{obs:well-organized-frozen} and the only vertex it can reach on $A_{start}$ is $a_{M_B(C)}$, which also leads to a well-organized configuration. By symmetry, the same goes for the moves between $B_{start}$ and $A$. It follows that the only possible bad moves are the moves of condition $1$, $2$ and $3$.
\end{proof}

We consider the smallest integer $t$ such that the move between $C_t$ and $C_{t+1}$ is a bad move.
Since $C_1$ is well-organized, $C_t$ is well-organized by definition of a bad move.
For brevity we set $i := M_A(C_t)$ and $j := M_B(C_t)$. Note that $i \leq j$ so there can be no
move from $B$ to $A$ unless $i = j$, in which case the move must be from $B_i$ to $A_i$ by Observation \ref{obs:well-organized-frozen}. By symmetry we can
thus always suppose that if the first bad move is a move between $A$ and $B$, then it is a move from $A$ to $B$. Furthermore, we can suppose that $i < 2k$ for otherwise
the configuration $C_t$ yields a multicolored independent set of size $k$ as shown in Section \ref{sec:well-organized-conf} and we are done. Using these symmetries and
Observation \ref{obs:well-nonwell} we can restrict ourselves to three cases: either the bad move is a move from $A$ to $B$, or it is a move from $A$ to $B_{end}$, or it is a move from $s_A$ to $B_{end}$. We denote these moves as bad moves of \emph{type 1}, \emph{type 2}, and \emph{type 3} respectively, and we denote the blue token making the bad move at time $t$ as the \emph{bad token}. Note that Observation \ref{obs:well-organized-frozen} ensures that if the bad move at time $t$ is of type $1$ or $2$, then the bad token is on $A_i$ in $C_t$. Since $i < 2k$, $C_t$ is well-organized, and the move at time $t$ is the first bad move of the sequence we have:
\begin{observation}\label{obs:token-sb-ct}
There is a red token on $s_B$ in $C_t$.
\end{observation}
The following observations give some more information about the configurations $C_t$ and $C_{t+1}$ that we obtain after the first bad move, depending on its type.
\begin{observation}\label{obs:bad-move-t1}
If the move at time $t$ is a bad move of type $1$, then $i := M_A(C_t)$ is odd. Furthermore, $(t: A_i \rightarrow B_q)$ with $q \geq i$.
\end{observation}
\begin{proof}
If $i$ is even, $i \geq 2$ and $A_{i-1}$ is a copy of $A_i$. Since $i \leq j$, Observation \ref{obs:wo-same-vertex-bis} ensures that there is a token on $A_{i-1}$, $A_{i}$ and $B_i$ on equivalent vertices, in which case the tokens on $A_{i-1}$ and $A_{i}$ cannot move to $B$, proving the first statement. The second statement is a direct consequence of Observation \ref{obs:well-organized-frozen}.
\end{proof}

\begin{observation}\label{obs:bad-move-t23}
If the move at time $t$ is a bad move of type $2$ or $3$, then $j := M_B(C_t) = 2k$.
\end{observation}
\begin{proof}
If $j < 2k$ then by definition of a well-organized configuration there are some blue tokens on $A_{start}$ and no token can move to $B_{end}$.
\end{proof}
Finally, the two following Observations follow from the fact that $C_t$ is well-organized:
\begin{observation}\label{obs:bad-move-t2}
If the move at time $t$ is a bad move of type $2$, then $(t: A_i \rightarrow b_{e,i})$.
\end{observation}
\begin{observation}\label{obs:bad-move-t3}
If the move at time $t$ is a bad move of type $3$, then $(t: s_{A} \rightarrow b_{e,p})$ with $p > i$.
\end{observation}

\subsubsection{Bad moves of type 1}\label{sec:bad-move-t1}
In this subsection, we suppose that the move at time $t$ is a bad move of type $1$. By Observation $\ref{obs:wo-same-vertex}$ we have $|A_p \cap C_t| = 1$ for every $p \leq M_A(C_t)$ and $|B_p \cap C_t| = 1$ for every $p \leq M_B(C_t)$. In this section $v^A_p$ (resp. $v^B_p)$ denote the only vertex of $|A_p \cap C_t|$ for $p \leq M_A(C_t)$ (resp. $|B_p \cap C_t|$ for $p \leq M_B(C_t)$. By Observation \ref{obs:bad-move-t1} we have $(t:A_i \rightarrow B_q)$ for some $q \geq i$. In the next lemma, we show that as long as no token moves from $B_q$ after time $t+1$, the blue tokens on $B$ and the red tokens on $B_{start}$ at time $t+1$ remain frozen.

\begin{lemma}\label{obs:tok-B-froz}
Let $t' \geq t+1$ such that no token has moved from $B_q$ between $C_{t+1}$ and $C_{t'}$. Then for any configuration between $C_{t+1}$ and $C_{t'}$ we have:
\begin{enumerate}
    \item for every $p \leq q$ there is a blue token on $v^B_p$.
    \item for every $p > i$ there is a red token on $b_{s,p}$.
\end{enumerate}
\end{lemma}
\begin{proof}
First, note that $C_{t+1}$ satisfies conditions $1$ and $2$ since $C_t$ is well-organized and the move between $C_t$ and $C_{t+1}$ is a bad move. Suppose for a contradiction that there exists a time $t+1 < \tau < t'$ such that for every $t+1 \leq \ell \leq \tau$ the configuration $C_{\ell}$ satisfies conditions $1$ and $2$ and that the configuration $C_{\tau + 1}$ does not. Then it must be that at time $\tau$, either a red token moves from $b_{s,p}$ for some $p > i$ or a blue token moves from $v^B_p$ for some $p < q$. Let us show that none of these moves is actually possible since $C_{\tau}$ satisfies conditions $1$ and $2$. \\
Suppose first that a red token moves from $b_{s,p}$ for $p > i$. Since $C_{\tau}$ satisfies condition $2$ the tokens on $b_{s,x}$ for $x > i+1$ are frozen and we have $p = i+1$. By construction the red token on $b_{s, i+1}$ can only move to $A_{i+1}$. Furthermore, $i$ is odd by Observation \ref{obs:bad-move-t1} and $A_{i+1}$ is a copy of $A_i$. By the choice of $\tau$ there is a blue token on $v^B_i$ and a red token on $N(v^A_{i+1})$ since $(\tau : v^A_{i+1} \rightarrow B_q)$. It follows that $A_{i+1}$ is fully dominated at time $\tau$ and that the red token on $b_{s, i+1}$ cannot move, a contradiction. \\
Suppose then that a blue token moves from $v^B_p$ for some $p < q$. Since there are tokens on $B_q$ and $p < q$, this token cannot move to $A_{start}$, and since condition $2$ is satisfied by $C_{\tau}$, it cannot move to $A_{end}$. Furthermore, since $C_{\tau}$ also satisfies condition $1$, no blue token can move to $A_x$ for $x \geq i+1$. We then have two sub-cases to consider :
\begin{enumerate}
    \item $p \leq i$: let $B_p'$ be the other copy of $B_p$ in $G$. By the choice of $\tau$ we have $B_p \cap C_{\tau} = \{v^B_p\}$, $B_{p'} \cap C_{\tau} = \{v^B_{p'}\}$ and by Observation \ref{obs:wo-same-vertex-bis} we have $v^B_p \sim v^B_{p'}$ hence the token on $v^B_p$ cannot move to $A$.
    \item $i+1 \leq p < q$: By the choice of $\tau$ we have $B_x \cap C_{\tau} = \{v^B_x\}$ with $v^B_x \sim v^A_x$ for every $x \leq i$ by Observation \ref{obs:wo-same-vertex-bis}. For such $x$, $v^A_x$ is the only vertex that is not dominated by the blue token on $A_x$, and since $C_t$ is an independent set we have $v^A_x \notin N(v^B_p)$. It follows that the blue token on $v^B_p$ cannot move to $A_x$ for $x \leq i$. Since $C_{\tau}$ satisfies condition $2$ it cannot move to $B_x$ for $x \leq i+1$, which concludes the proof.
\end{enumerate}
\end{proof}

So as long as there are two tokens on $B_q$ some tokens remain frozen and cannot reach the targeted independent set. Hence one of the two tokens on $B_q$ has to move again at some point in the reconfiguration sequence. The following Observation shows that one of the tokens on $B_q$ necessarily moves back to $A_i$.

\begin{observation}\label{obs:ct+1}
There exists $t' \geq t+1$ such that $(t': B_q \rightarrow A_i)$.
\end{observation}
\begin{proof}
To reach the target configuration, every token on $B_{start}$ must move at least one time. By Lemma \ref{obs:tok-B-froz}.1, the tokens on $b_{s,p}$ for $p > i$ cannot move as long as there are two tokens on $B_q$. It follows that one of these token has to move at a time $t' \geq t+1$. Let $u \in B_q$ be the vertex such that $(t: v^A_i \rightarrow u)$: note that $u \notin N(v^A_p)$ for any $p < i$. Then by Lemma \ref{obs:tok-B-froz}.1 there can be no move from $B_q$ to $A_p$ for $p < i$ and by \ref{obs:tok-B-froz}.2 there can be no move from $B_q$ to $A_{p}$ for $p > i+1$ at time $\tau$. Furthermore, Lemma \ref{obs:tok-B-froz}.1 also ensures that there can be no move from $B_q$ to $a_{s,p}$ for $p < q$, and since there are two tokens on $B_q$ at time $\tau$, none of them can move to $b_{s,q}$.
Thus, $(\tau: B_q \rightarrow A_i)$ is the only possible move at time $\tau$.
\end{proof}

In other words, the bad move at time $t$ is in some sense "canceled" at time $t'$. Note that, however, it is not necessarily the red token that moves at time $t'$: in the particular case where $q = j = i$, the blue token on $B_q$ can move to $A_i$, switching role with the blue token. The next lemma shows that in-between $t$ and $t'$ every token has a very restricted pool of possible moves and remains locked in the closed neighborhood of the token it lies on in $C_t$.

\begin{lemma}\label{lem:tok-A-froz}
Let $t' \geq t+1$ be the first time after $t$ such that $(t': B_q \rightarrow A_i)$. Then any configuration $C_{\ell}$ with $t+1 \leq \ell \leq t'$ satisfies the following conditions:
\begin{enumerate}
        \item For every even $p < i$, there is either a red token on $b_{s,p}$ or a red token on $v^A_p$ or a red token on $b_{e,p}$.
        \item For every odd $p < i$, there is either a red token on $b_{s,p}$, or a red token on $v^A_p$, or a red token on $N(v^A_p) \cap B$, or red token on $b_{e,p}$.
        \item For every $p > q$ there is either a blue token on $a_{s,p}$ or a blue token on $B_p$.
\end{enumerate}
\end{lemma}
\begin{proof}
Let us first show that $C_{t+1}$ satisfies conditions $1$ to $3$. Since $C_t$ is well-organized and since $(t: A_i \rightarrow B_q)$ there is a red token on $v^A_p$ for every $p < i$ thus conditions $1$ and $2$ are satisfied. Furthermore, there is a blue token on $B_p$ for every $q < p \leq M_B(C)$ and a blue token on $a_{s,p}$ for every $M_B(C) < p \leq 2k$ and condition $3$ is satisfied by $C_{t+1}$. Let us now prove that these conditions are satisfied by any configuration between times $t$ and $t'$. Suppose otherwise and let $\tau$ be the first time after $t+1$ such that $C_{\tau}$ does not satisfy one of the three conditions. Note that by Lemma \ref{obs:tok-B-froz}.1 , we know that for any $t+1 \leq \ell \leq t'$ there is a blue token on $v^B_p$ for every $p < q$ in $C_{\ell}$ and a red token on $b_{s,p}$ for every $p > i$.

\begin{enumerate}
\item \textbf{$C_{\tau}$ does not satisfy condition $1$.} Since $C_{\tau-1}$ satisfies the three conditions, there exist exactly one even integer $p_0 < i$ for which condition $1$ is not satisfied in $C_{\tau}$.
\begin{enumerate}
    \item Suppose first there is a token on $b_{s,p_0}$ in $C_{\tau-1}$. Since there is a token on $b_{s,p}$ for every $p > i$, this token cannot move to $A_{end}$ nor $B_{p}$ for any such $p$, and there is no token on $A_p$ for any $p > p_0$ in $C_{\tau-1}$. Then, since conditions $1$ and $2$ are satisfied by $C_{\tau-1}$, there must be a red token on $\{b_{s,p}, b_{e,p}, v^A_p\} \cup N(v^A_p) \cap B$ for every $p_0 < p \leq i$. The token on $b_{s,p_0}$ then has to move to $A_{p_0}$ a time $\tau-1$, and since there is a blue token on $v^B_{p_0}$, the only vertex it can move to is $v^A_{p_0}$. But then $C_{\tau}$ satisfies condition $1$, a contradiction.
    \item Suppose then that there is a red token on $b_{e,p_0}$ in $C_{\tau-1}$: since $N(b_{e,p_0}) \cap A$ = $N(b_{s,p_0}) \cap A$, one can easily see that the only vertex this token can move to is also $v^A_{p_0}$, again leading to a contradiction.
    \item Finally suppose that there is a red token on $v^A_{p_0}$ in $C_{\tau-1}$. Then there can be no token on $b_{s, p_0-1}$ nor on $b_{e, p_0-1}$. Furthermore - recall that since
    $p_0$ is even $A_{p_0}$ and $A_{p_0-1}$ are copies of the same set - there can be no token in $N(v^A_{p_0-1}) \cap B$ in $C_{\tau-1}$. Since condition $2$ is satisfied for $p_0-1$, there must then be a token on $v^A_{p_0-1}$. It follows that the token on $A_{p_0}$ can only move to $b_{s,p_0}$ or $b_{e, p_0}$ and condition $1$ is satisfied by $C_{\tau}$, a contradiction.
\end{enumerate}

\item \textbf{$C_{\tau}$ does not satisfy condition $2$.} As in case $1$, there exists exactly one odd integer $p_0 < i$ for which condition $2$ is not satisfied in $C_{\tau}$. If there is a token on $b_{s,p_0}$ or on $b_{e,p_0}$ in $C_{\tau-1}$ we obtain a contradiction using the same arguments (which do not make use of the parity of $p_0$) than in case $1.a$ and $1.b$ respectively. Two cases remain to be considered:
\begin{enumerate}
    \item Suppose that there is a token on $v^A_{p_0}$ in $C_{\tau-1}$. Then there can be no token on $b_{s, p_0-1}$ nor on $b_{e, p_0-1}$ and since $C_{\tau-1}$ satisfies condition $1$ ($p_0-1$ is even), there must be a token on $v^A_{p_0-1}$ in $C_{\tau-1}$. It follows that the token on $A_{p_0}$ can either move to $b_{s,p_0}$, $b_{e,p_0}$ or to $B$, and $C_{\tau}$ satisfies condition $2$.
    \item Finally suppose there is a red token on $N(v^A_{p_0}) \cap B $ in $C_{\tau-1}$. Let $p_1$ be such that this red token is on $B_{p_1}$. Then by construction there can be no token on $a_{s,p_1}$ in $C_{\tau-1}$ and since condition $3$ is satisfied by $C_{\tau-1}$ there is also a blue token on $B_{p_1}$ in $C_{\tau-1}$. It follows that there are two tokens on $B_{p_1}$ in $C_{\tau-1}$ and that these tokens cannot move to $A_{start}$. Furthermore, since $i < 2k$ we have $B_{start} \cap C_{\tau} \neq \emptyset$, so the red token on $B_{p_1}$ cannot move to $A_{end}$ and must move to $B$ at time $\tau-1$. Let us show it can only move back to $v^A_{p_0}$. By Lemma \ref{obs:tok-B-froz}.1 this token can only move to $v^A_{p}$ for some $p < i$. But since $C_{\tau-1}$ satisfies condition $1$ and $2$, we have that for any $p \neq p_0$, there is a red token on $\{b_{s,p}, b_{e,p}, v^A_{p}\} \cup N(v^A_p) \cap B$. It follows that the only vertex of $A$ this token can move to is $v^A_{p_0}$ and condition $2$ is satisfied by $C_{\tau}$
\end{enumerate}
\item \textbf{$C_{\tau}$ does not satisfy condition $3$.} As in the previous cases there exists exactly one integer $p_0 > q$ for which condition $3$ is not satisfied in $C_{\tau}$.
\begin{enumerate}
    \item Suppose first there is a blue token on $a_{s,p_0}$ in $C_{\tau-1}$. If $p_0 = 2k$ this token can only move to $B_{2k}$ and we are done. Otherwise, there can be no token on $B_{p_0+1}$ in $C_{\tau-1}$ and since this configuration satisfies condition $3$, there must then be a token on $a_{s,p_0+1}$. It follows that the blue token on $a_{s, p_0}$ can only move to $B_{p_0}$ and we obtain a contradiction.
    \item Suppose then that there is a blue token on $B_{p_0}$ in $C_{\tau-1}$. Since there are still tokens on $B_{start}$ this token cannot go to $A_{end}$ nor to any $A_p$ for $p > i$. Furthermore by Lemma \ref{obs:tok-B-froz}.1, the only vertex on $A_p$ that is not dominated by tokens on $B$ is $v^A_{p}$ for any $p \leq i$. But since $C_{\tau-1}$ satisfies condition $1$ and $2$, there is a red token on $\{b_{s,p}, b_{e,p}\} \cup B$ that dominates this vertex. It follows that the blue token on $B_{p_0}$ can only move to $A_{start}$. Furthermore since there is a token on $B_{p_0}$ there can be no token on $a_{s,p}$ for $p \leq p_0$ and since $C_{\tau-1}$ satisfies condition $3$ there must be a blue token on $B_{p}$ for any $p \leq p_0$. It follows that the blue token on $B_{p_0}$ in $C_{\tau-1}$ can only move to $a_{s,p_0}$ and that condition $3$ is satisfied by $C_{\tau}$, which concludes the proof.
\end{enumerate}
\end{enumerate}

\end{proof}

Furthermore, up to removing a move from the sequence we have the following:
\begin{observation}\label{obs:lock-token-t1}
Let $t' \geq t+1$ such that no token has moved from $B_q$ between $C_{t+1}$ and $C_{t'}$. Then for any configuration between $C_{t+1}$ and $C_{t'}$ there is a token on $\{s_A, e_B\}$.
\end{observation}
\begin{proof}
Since the move at time $t$ is the first bad move, there is a token on $\{s_A, e_B\}$ at time $t$. If there is a token on $e_B$, there can be no token on $A_{start}$ and by Lemma \ref{obs:tok-B-froz}.3 there must be blue token on $B_p$ for every $p \leq 2k$ so there can be no move from $e_B$ to $A_{start}$. Suppose there exists $\tau > t$ such that $(\tau: s_A \rightarrow B_{e,p})$ for some $p$. By Lemma \ref{lem:tok-A-froz} and \ref{obs:tok-B-froz}, this token cannot move to $A$ before time $t'+1$. But then we can replace the move at time $\tau$ by $(\tau: s_A \rightarrow e_B)$: since $N(e_B) \subseteq N(b_{e,p})$ all the moves between time $\tau$ and $t'+1$ remain valid.
\end{proof}

Let us now consider the configurations $C_{t'}$ and $C_{t'+1}$. We know that $(t': B_q \rightarrow A_i)$ and in particular there can be no token on $b_{s,p}$ nor $b_{e,p}$ for any $p \leq i$ in $C_{t'+1}$. Since configuration $C_{t'+1}$ satisfies condition $1$ of Lemma \ref{lem:tok-A-froz} we have that for any even $p  < i$ there is a red token on $v^A_p$, and since $v^A_{p} \sim v^A_{p-1}$ there cannot be any red token on $N(v^A_{p-1}) \cap B$. Then by condition $2$ of Lemma \ref{lem:tok-A-froz} there is necessarily a token on $v^A_{p-1}$. Furthermore, by Lemma \ref{obs:tok-B-froz}.2 there is a red token on $b_{s,p}$ for every $p > i$, and by Lemma \ref{obs:tok-B-froz}.1 there is a blue token on $B_p$ for every $p < q$ in $C_{t'+1}$. Condition $3$ of Lemma \ref{lem:tok-A-froz} ensures that there is a blue token on $\{b_{s,p}\} \cup B_p$ for every $p > q$ in $C_{t+1}$. Furthermore, Observation \ref{obs:lock-token-t1} and Observation \ref{obs:token-sb-ct} ensure that there is token on $\{s_A, e_B\}$ and a token on $\{s_B, e_A\}$. Finally, there are two tokens on $B_q$ at time $t'$ and one of these moves to $A_i$, which ensures that $C_{t'+1}$ is well-organized.

In the considered shortest sequence $S$, the token that moves from $B$ to $A$ at time $t'$ moves at least three times before we reach the well-organized configuration $C_{t'+1}$, a contradiction with the choice of $S$ by Lemma \ref{obs:min-move-wo}.

\subsubsection{Bad moves of type $2$ and $3$}\label{sec:bad-move-t2-t3}
The proof for bad moves of type $2$ and $3$ follows similar reasoning as for type $1$. We first show that as long as the bad token does not move after time $t+1$, a large part of the other tokens remain frozen. We then show that the bad token has to move again after time $t+1$ and that we subsequently either obtain a well-organized configuration or cancel a bad move. By Observation \ref{obs:bad-move-t2} and \ref{obs:bad-move-t3} we have either $(t: A_i \rightarrow b_{e,q}$) or $(t: s_A \rightarrow b_{e,q})$ for some $q \geq i$. Note that in the particular case of a bad move a type $2$ we have $q = i$. By Observation \ref{obs:bad-move-t23}, there is exactly one blue token on $B_p$ for every $1\leq p \leq 2k$ in $C_t$. We denote by $v^B_{p}$ the only vertex of  $I_t \cap B_p$ and by $v^A_p$ the copy of $v_p^B$ in $A_p$. \\
Let $p \leq 2k$ be odd. Recall that by construction, if $v^B_p \sim v^B_{p+1}$ then the blue tokens on $\{v^B_p,  v^B_{p+1}\}$ cannot move to $A$, and no other token can move to $(A_p \cup A_{p+1}) - \{v_p^A, v_{p+1}^A\}$. If $v_p^B \nsim v_{p+1}^B$, no token can move $A_p \cup A_{p+1}$ except for the tokens on $\{v_p^B, v_{p+1}^B\}$.
Let us first show that there necessarily exists a time $t' > t$ at which the bad token moves again:
\begin{observation}\label{obs:token-sb-frozen}
Let $t_1 > t$ be such that for any $t \leq \tau \leq \ell$, $b_{e,q} \in C_{\tau}$. Then for any $t \leq \tau \leq \ell$, $s_B \in C_{\tau}$.
\end{observation}
\begin{proof}
As long as there is a token on $b_{e,q}$, no token on $b_{s,p}$ for $p \geq q$ can move to $A$. Since by Observation \ref{obs:token-sb-ct}, $s_B \in C_{t}$ and since $b_{s,p} \in C_t$ for $p > q$, these token are frozen as long as there is a token on $b_{e,q}$.
\end{proof}

\begin{observation}\label{obs:type-2-eb}
If the move at time $t$ is a bad move of type $2$, then there is a blue token on $e_B$ in $C_t$ and this token cannot move before the bad token moves again.
\end{observation}

In order to show that the configuration $C_{t'+1}$ is well-organized, we need a lemma similar to Lemma \ref{lem:tok-A-froz}:
\begin{lemma}\label{lem:tok-A-froz-bis}
Let $t' \geq t+1$ be the smallest integer such that the move between $C_{t'}$ and $C_{t'+1}$ is a move of the bad token. Then any configuration $C_{\ell}$ with $t+1 \leq \ell \leq t'$ satisfies the following conditions:
\begin{enumerate}
        \item If $p < q$ is even, there is either a red token on $b_{s,p}$ or a red token on $v^A_{p}$, or a red token on $b_{e,p}$.
        \item If $p < q$ is odd, there is either a red token on $b_{s,p}$, or a red token on $v^A_p$, or a red token in $N(v^A_p) \cap B$, or red token on $b_{e,p}$.
        \item For every $p \leq 2k$ there is a blue token on $v^B_p$.
        \item For every $p \geq q$ there is a red token on $b_{s,p}$ and there is a red token on $s_B$.
\end{enumerate}
\end{lemma}
\begin{proof}
As the configuration $C_t$ is well-organized and $(t: A_i \rightarrow B_{start})$, conditions $1$ to $4$ are satisfied by configuration $C_{t+1}$. As for the proof of Lemma \ref{lem:tok-A-froz}, we suppose for a contradiction that there exist a time $t+1 < \tau < t'$ such that $C_{\tau}$ satisfies conditions $1$ to $4$ and $C_{\tau+1}$ does not. We consider the smallest such time $\tau$.

\begin{enumerate}
    \item \textbf{$C_{\tau+1}$ does not satisfy condition $1$.} Since $C_{\tau}$ satisfies the four conditions, there exist exactly one even integer $p_0 < q$ for which condition $1$ is not satisfied in $C_{\tau+1}$.
    \begin{enumerate}
        \item Suppose there is a red token on $b_{s,p_0}$ in $C_{\tau}$. Since $C_{\tau}$ satisfies condition $4$, this token cannot move to $B_{p}$ for any $p > q$ nor it can move to $A_{end}$. Since there is a token on $b_{e,q}$, it cannot move to $A_q$ either. So we must have $(\tau: b_{s,p_0} \rightarrow A_x)$ for some $p_0 \leq x < q$. Suppose w.l.o.g that $x$ is even: by condition $3$, there is a token on both $v^B_x$ and $v^B_{x-1}$ with $v^B_x \sim v^B_{x-1}$, thus we have $(\tau: b_{s,p_0} \rightarrow v^A_x)$. But then if $x \neq p_0$, condition $1$ and $2$ ensure that there is either a token on $v^A_x$, $b_{e,x}$, $b_{s,x}$ or on $N(v^A_x) \cap B$, a contradiction.
        \item Suppose there is a red token on $v^A_{p_0}$ in $C_{\tau}$. Since $C_{\tau}$ satisfies condition $2$, there is also a token on $v^A_{p_0-1}$ and by condition $3$ we have $v^A_{p_0-1} \sim v^A_{p_0-1}$. So the token on $v^A_{p_0}$ can either move to $b_{e, p_0}$ or $b_{s, p_0}$ since any other of its neighbors in $B_{end}$ or $B_{start}$ dominates $A_{p_0 - 1}$, and condition $1$ remains satisfied.
        \item Suppose there is a red token on $b_{e, p_0}$ in $C_{\tau}$. Since there is a token on $b_{e, q}$, the token on $b_{e, p_0}$ cannot move to $A_{end}$, and it must move to $A$. Since the vertices $b_{s, p_0}$ and $b_{e, p_0}$ share the same neighborhood in $A$, we can apply the same arguments as for case $1.a$ showing that $(\tau: b_{e, p_0} \rightarrow v^A_{p_0})$, and condition $1$ remains satisfied.

    \end{enumerate}
    \item \textbf{$C_{\tau+1}$ does not satisfy condition $2$.} Since $C_{\tau}$ satisfies the four conditions, there exist exactly one odd integer $p_0 < q$ for which condition $2$ is not satisfied in $C_{\tau+1}$. First note that the proof for case $1$ do not make use at any point of the parity of $p_0$. Hence if there is a token on $b_{s,p_0}$, $v^A_{p_0}$ or on $b_{e,p_0}$  case $1.a$, $1.b$ and $1.c$ apply respectively. Only the case where there is a red token on $N(v^A_{p_0}) \cap B$ in $C_{\tau}$ remains to be considered. Let $u$ denote the vertex on which the token is. Since condition $3$ is satisfied by $C_{\tau}$, this token cannot move to $A_{start}$ nor
    $A_{end}$, so we have $(\tau: N(u) \rightarrow v^A_x)$ for some $x$. Suppose that $A_x$ and $A_{p_0}$ are not copies of the same set. First note that $x$ must be odd, for otherwise there must be a token on $b_{s,x}$ or $b_{e,x}$ at time $\tau$ by condition $1$ and it is not possible to move from $B$ to $A_x$ at time $\tau$. Furthermore we must have $x < p$: if not, then by condition $1$ there must be a token on $b_{s,p_0+1}$ or $b_{e,p_0+1}$ which dominates $A_x$ since $p_0$ is odd and since there can be no token on $v^A_{p_0+1}$.\\
    There can be no move from $v^A_{p_0}$ to $u$ at any time between $t$ and $\tau$. Suppose otherwise and consider a time $t_0$ such that $t < t_0 < \tau$ at which such a move occurs: in $C_{t_0}$ there is a token on $v^A_{p_0}$ so there cannot be any token on $b_{e,x+1}$ nor $b_{s,x+1}$ hence there is a token on $v^A_{x+1} \in N(u)$ by condition $1$. It follows that, since condition $2$ is satisfied for any time $t \leq \tau$, there must be a token either on $b_{s,p_0}$ or $b_{e,p_0}$ or on a vertex of $N(v^A_{p_0}) \cap B$ which is distinct from $u$. But then after the token on $u$ moves to $v^A_x$ at time $\tau$ condition $2$ is still satisfied, a contradiction.

    \item \textbf{$C_{\tau+1}$ does not satisfy condition $3$.} Since $C_{\tau}$ satisfies condition $3$, there exists a unique $p_0 \leq 2k$ such that at time $\tau$ a blue token moves from $v^B_{p_0}$. Since condition $3$ is satisfied and since there is a token on $b_{e,q}$, this token cannot move to $A_{end}$ nor $A_{start}$. So we can suppose that $(\tau: v^B_{p_0} \rightarrow A_x)$ for some $x$ odd without loss of generality. If $v^B_x \nsim v^B_{x+1}$ then there must be a red token on $b_{s,x+1}$ or $b_{e,x+1}$ by condition $1$ since there can be no token on $v^A_{x+1}$. So it must be that $v^B_x \sim v^B_{x+1}$: but then again by condition $1$ and $2$ there must be either a red token on $b_{s,x}$, $b_{e,x}$ or on $N(v^A_x)\cap B$ and it follows that no blue token can move to $A_x$.

    \item \textbf{$C_{\tau+1}$ does not satisfy condition $4$.} As long as there are some tokens on $B_{start}$, the red token on $s_B$ cannot move, so there exists a unique $p_0 \geq q$ such that at time $\tau$ a red token moves from $b_{s,p_0}$. Since there is a token on $s_B$ this token cannot move to $A_{start}$ and thus can only move to $A$. By construction it can only move to $A_x$ for some $x \geq q$ and any such set is dominated by the token on $b_{e,q}$, a contradiction.
\end{enumerate}
\end{proof}
As long as the bad token does not move again after time $t$, condition $4$ of Lemma \ref{lem:tok-A-froz-bis} ensures that the red tokens on $B_{start} \cup \{s_B\}$ remain frozen. So there must exist a time $t' > t$ such that the bad token moves at time $t'$. The following observation actually show that this token moves back to the position it had in $C_t$:
\begin{observation}\label{obs:bad-token-moves-back}
Let $t' > t$ denote the time at which the bad token moves again. We have the following:
\begin{enumerate}
    \item $(t': b_{e,q} \rightarrow v^A_i)$ if the move at time $t$ is a bad move of type $2$.
    \item $(t': b_{e,q} \rightarrow s_A)$ if the move at time $t$ is a bad move of type $3$.
\end{enumerate}
\end{observation}
\begin{proof}
We prove the two statements separately:
\begin{enumerate}
    \item \textbf{The move at time $t$ is a bad move of type 2.} By Lemma \ref{lem:tok-A-froz-bis}.$3$ there is a blue token on $B_p$ for every $p \leq 2k$ at time $t'$ so the bad token cannot move to $A_{start}$. Furthermore by Observation \ref{obs:type-2-eb} there is a blue token on $e_B$ at time $t'$ so it cannot move to $s_A$ either. By Observation \ref{obs:bad-move-t2} we have $q = i$ and thus $(t': b_{e,q} \rightarrow A_x)$ for some $x \geq i$. By Lemma \ref{lem:tok-A-froz-bis}.4 there are red tokens on $b_{s,p}$ for every $p \geq i$ so it must be that $x = i$. Finally Lemma \ref{lem:tok-A-froz-bis}.3 ensures that there is a token on $v^B_i$ in $C_{t'}$ and since $v^A_i \sim v^B_i$ it follows that $(t': b_{e,q} \rightarrow v^A_i$) is the only possible move for the bad token at time $t'$.
    \item \textbf{The move at time $t$ is a bad move of type 3.} As for the previous case, Lemma \ref{lem:tok-A-froz-bis}.$3$ ensures that the bad token cannot move to $A_{start}$. By Observation \ref{obs:bad-move-t3} we have $q > i$ and by Lemma \ref{lem:tok-A-froz-bis}.4 there is a token on $b_{s,q}$ so the bad token cannot move to $A$. It follows that $(t': b_{e,q} \rightarrow s_A)$ is the only possible move for the bad token at time $t'$.
\end{enumerate}
\end{proof}

The following Observation allows us to conclude about the bad moves of type 2:
\begin{observation}\label{obs:bad-move-t2-conclusion}
Let $t' \geq t+1$ denote the time at which the bad token moves again. If the move at time $t$ is a bad move of type $2$, then the configuration $C_{t'+1}$ is well-organized.
\end{observation}
\begin{proof}
By Observation \ref{obs:bad-token-moves-back}.1, we have that $(t': b_{e,i} \rightarrow v^A_i)$. Since $q = i$ by Observation \ref{obs:bad-move-t2}, Lemma \ref{lem:tok-A-froz-bis} ensures that there is a red token on $b_{s,p}$ for every $p \geq i$ and a blue token on $B_x$ for every $x \leq 2k$ a time $t'$. It remains to show that there is a red token on $A_y$ for every $y \leq i$ to obtain a well-organized configuration. Since there is a token on $A_i$ in $C_{t'+1}$ there can be no token on $b_{s,p}$ nor $b_{e,p}$ for any $p \leq i$ and condition $1$ of Lemma \ref{lem:tok-A-froz-bis} then ensures that there is a red token on $v^A_p$ for every even $p \leq i$. Since furthermore $v^A_p \sim v^A_{p+1}$ for every odd $p < i$, there can be no token on $N(v^A_{p+1}) \cap B = N(v^A_p) \cap B$ for any such $p$, and condition $2$ of Lemma \ref{lem:tok-A-froz-bis} ensures that there is a red token on $v^A_p$ for every odd $p < i$.
\end{proof}
As for bad moves of type $1$, there is a token that moves at least three times to reach the well-organized configuration $C_{t'+1}$, a contradiction with Observation \ref{obs:min-move-wo}.

It remains to check the case of bad moves of type 3. If $(t: s_A \rightarrow b_{e,q})$ we proceed as follows: we replace the move at time $t$ by the move $(t: s_A \rightarrow e_B)$ and the move at time $t'$ by $(t': e_B \rightarrow s_A)$. Since $N(e_B) \subseteq N(b_{e,q})$, the moves at times $t, t+1, \ldots, t'$ remain valid. Furthermore by Observation \ref{obs:bad-token-moves-back}.2 we obtain the same independent set at time $t'+1$. Although the modified sequence is not shorter, it does not contain any bad move of type $3$: either we obtain a well-organized configuration at time $t'+1$ and we are done, or there is a bad move of type $1$ or $2$ between time $t+1$ and $t'$, in which case one of the previous cases apply. It follows that the sequence contains no bad move of type 3. The proof of Lemma \ref{lem:no-bad-move-shortest} is now straightforward:
\begin{proof}
Let $S$ be a shortest reconfiguration sequence from $I_s$ to $I_e$. In section \ref{sec:bad-move-t1} we showed that $S$ contains no bad move of type $1$, and we showed that $S$ contains no bad moves of type $2$ or $3$ in section \ref{sec:bad-move-t2-t3}. Then by Observation \ref{obs:well-nonwell} it follows that $S$ contains no bad move.
\end{proof}

%% file: connexion_equiv.tex
\tikzstyle{ipe stylesheet} = [
  ipe import,
  even odd rule,
  line join=round,
  line cap=butt,
  ipe pen normal/.style={line width=0.4},
  ipe pen heavier/.style={line width=0.8},
  ipe pen fat/.style={line width=1.2},
  ipe pen ultrafat/.style={line width=2},
  ipe pen normal,
  ipe mark normal/.style={ipe mark scale=3},
  ipe mark large/.style={ipe mark scale=5},
  ipe mark small/.style={ipe mark scale=2},
  ipe mark tiny/.style={ipe mark scale=1.1},
  ipe mark normal,
  /pgf/arrow keys/.cd,
  ipe arrow normal/.style={scale=7},
  ipe arrow large/.style={scale=10},
  ipe arrow small/.style={scale=5},
  ipe arrow tiny/.style={scale=3},
  ipe arrow normal,
  /tikz/.cd,
  ipe arrows, %
  <->/.tip = ipe normal,
  ipe dash normal/.style={dash pattern=},
  ipe dash dotted/.style={dash pattern=on 1bp off 3bp},
  ipe dash dashed/.style={dash pattern=on 4bp off 4bp},
  ipe dash dash dotted/.style={dash pattern=on 4bp off 2bp on 1bp off 2bp},
  ipe dash dash dot dotted/.style={dash pattern=on 4bp off 2bp on 1bp off 2bp on 1bp off 2bp},
  ipe dash normal,
  ipe node/.append style={font=\normalsize},
  ipe stretch normal/.style={ipe node stretch=1},
  ipe stretch normal,
  ipe opacity 10/.style={opacity=0.1},
  ipe opacity 30/.style={opacity=0.3},
  ipe opacity 50/.style={opacity=0.5},
  ipe opacity 75/.style={opacity=0.75},
  ipe opacity opaque/.style={opacity=1},
  ipe opacity opaque,
]
\definecolor{red}{rgb}{1,0,0}
\definecolor{blue}{rgb}{0,0,1}
\definecolor{green}{rgb}{0,1,0}
\definecolor{yellow}{rgb}{1,1,0}
\definecolor{orange}{rgb}{1,0.647,0}
\definecolor{gold}{rgb}{1,0.843,0}
\definecolor{purple}{rgb}{0.627,0.125,0.941}
\definecolor{gray}{rgb}{0.745,0.745,0.745}
\definecolor{brown}{rgb}{0.647,0.165,0.165}
\definecolor{navy}{rgb}{0,0,0.502}
\definecolor{pink}{rgb}{1,0.753,0.796}
\definecolor{seagreen}{rgb}{0.18,0.545,0.341}
\definecolor{turquoise}{rgb}{0.251,0.878,0.816}
\definecolor{violet}{rgb}{0.933,0.51,0.933}
\definecolor{darkblue}{rgb}{0,0,0.545}
\definecolor{darkcyan}{rgb}{0,0.545,0.545}
\definecolor{darkgray}{rgb}{0.663,0.663,0.663}
\definecolor{darkgreen}{rgb}{0,0.392,0}
\definecolor{darkmagenta}{rgb}{0.545,0,0.545}
\definecolor{darkorange}{rgb}{1,0.549,0}
\definecolor{darkred}{rgb}{0.545,0,0}
\definecolor{lightblue}{rgb}{0.678,0.847,0.902}
\definecolor{lightcyan}{rgb}{0.878,1,1}
\definecolor{lightgray}{rgb}{0.827,0.827,0.827}
\definecolor{lightgreen}{rgb}{0.565,0.933,0.565}
\definecolor{lightyellow}{rgb}{1,1,0.878}
\definecolor{black}{rgb}{0,0,0}
\definecolor{white}{rgb}{1,1,1}
\begin{tikzpicture}[ipe stylesheet, xscale=0.7, yscale=0.7]
  \draw[shift={(159.516, 448.484)}, rotate=89.981]
    (0, 0) rectangle (32, -192);
  \draw[shift={(159.569, 576.484)}, rotate=89.981]
    (0, 0) rectangle (32, -192);
  \draw[shift={(255.516, 448.452)}, rotate=89.981]
    (0, 0)
     -- (32, 0);
  \draw[shift={(255.569, 576.452)}, rotate=89.981]
    (0, 0)
     -- (32, 0);
  \pic
     at (175.5215, 464.4786) {ipe disk};
  \pic
     at (207.5215, 464.468) {ipe disk};
  \pic
     at (239.5215, 464.4574) {ipe disk};
  \pic
     at (271.5215, 464.4467) {ipe disk};
  \pic
     at (303.5215, 464.4361) {ipe disk};
  \pic
     at (335.5215, 464.4255) {ipe disk};
  \pic
     at (335.5746, 592.4255) {ipe disk};
  \pic
     at (303.5746, 592.4361) {ipe disk};
  \pic
     at (271.5746, 592.4467) {ipe disk};
  \pic
     at (239.5746, 592.4574) {ipe disk};
  \pic
     at (207.5746, 592.468) {ipe disk};
  \pic
     at (175.5746, 592.4786) {ipe disk};
  \draw[red, ipe pen heavier]
    (176, 464)
     -- (207.574, 592.468)
     -- (207.574, 592.468);
  \draw[red, ipe pen heavier]
    (176, 464)
     -- (239.574, 592.458)
     -- (239.574, 592.458);
  \draw[red, ipe pen heavier]
    (176, 464)
     -- (303.574, 592.437)
     -- (303.574, 592.437);
  \draw[red, ipe pen heavier]
    (176, 464)
     -- (335.574, 592.426)
     -- (335.574, 592.426);
  \draw
    (208, 464)
     -- (175.574, 592.479)
     -- (175.574, 592.479);
  \draw
    (208, 464)
     -- (239.574, 592.458)
     -- (239.574, 592.458);
  \draw
    (208, 464)
     -- (271.574, 592.447)
     -- (271.574, 592.447);
  \draw
    (208, 464)
     -- (335.574, 592.426)
     -- (335.574, 592.426);
  \draw
    (240, 464)
     -- (175.574, 592.479)
     -- (175.574, 592.479);
  \draw
    (240, 464)
     -- (207.574, 592.468)
     -- (207.574, 592.468);
  \draw
    (240, 464)
     -- (271.574, 592.447)
     -- (271.574, 592.447);
  \draw
    (240, 464)
     -- (303.574, 592.437)
     -- (303.574, 592.437);
  \draw[red, ipe pen heavier]
    (272, 464)
     -- (207.574, 592.468)
     -- (207.574, 592.468);
  \draw[red, ipe pen heavier]
    (272, 464)
     -- (239.574, 592.458)
     -- (239.574, 592.458);
  \draw[red, ipe pen heavier]
    (272, 464)
     -- (303.574, 592.437)
     -- (303.574, 592.437);
  \draw[red, ipe pen heavier]
    (272, 464)
     -- (335.574, 592.426)
     -- (335.574, 592.426);
  \draw
    (304, 464)
     -- (175.574, 592.479)
     -- (175.574, 592.479);
  \draw
    (304, 464)
     -- (239.574, 592.458)
     -- (239.574, 592.458);
  \draw
    (304, 464)
     -- (271.574, 592.447)
     -- (271.574, 592.447);
  \draw
    (304, 464)
     -- (335.574, 592.426)
     -- (335.574, 592.426);
  \draw
    (336, 464)
     -- (175.574, 592.479)
     -- (175.574, 592.479);
  \draw
    (336, 464)
     -- (207.574, 592.468)
     -- (207.574, 592.468);
  \draw
    (336, 464)
     -- (271.574, 592.447)
     -- (271.574, 592.447);
  \draw
    (336, 464)
     -- (303.574, 592.437)
     -- (303.574, 592.437);
  \node[ipe node, font=\large]
     at (203.498, 432.468) {$A_{2p-1}$};
  \node[ipe node, font=\large]
     at (287.498, 432.437) {$A_{2p}
$};
  \node[ipe node, font=\large]
     at (291.591, 616.437) {$B_{2p}$};
  \node[ipe node, font=\large]
     at (203.591, 616.466) {$B_{2p-1}$};
  \node[ipe node, font=\Large]
     at (171.508, 452.479) {$u$};
  \node[ipe node, font=\Large]
     at (203.508, 452.468) {$v$};
  \node[ipe node, font=\Large]
     at (235.508, 452.458) {$w$};
  \node[ipe node, font=\Large]
     at (267.508, 452.479) {$u$};
  \node[ipe node, font=\Large]
     at (299.508, 452.468) {$v$};
  \node[ipe node, font=\Large]
     at (331.508, 452.458) {$w$};
  \node[ipe node, font=\Large]
     at (171.508, 596.479) {$u$};
  \node[ipe node, font=\Large]
     at (203.508, 596.468) {$v$};
  \node[ipe node, font=\Large]
     at (235.508, 596.457) {$w$};
  \node[ipe node, font=\Large]
     at (267.508, 596.479) {$u$};
  \node[ipe node, font=\Large]
     at (299.508, 596.468) {$v$};
  \node[ipe node, font=\Large]
     at (331.508, 596.457) {$w$};
  \draw[red]
    (172, 468) rectangle (180, 460);
  \draw[red]
    (268, 468) rectangle (276, 460);
\end{tikzpicture}

%% file: connexion_not_equiv.tex
\tikzstyle{ipe stylesheet} = [
  ipe import,
  even odd rule,
  line join=round,
  line cap=butt,
  ipe pen normal/.style={line width=0.4},
  ipe pen heavier/.style={line width=0.8},
  ipe pen fat/.style={line width=1.2},
  ipe pen ultrafat/.style={line width=2},
  ipe pen normal,
  ipe mark normal/.style={ipe mark scale=3},
  ipe mark large/.style={ipe mark scale=5},
  ipe mark small/.style={ipe mark scale=2},
  ipe mark tiny/.style={ipe mark scale=1.1},
  ipe mark normal,
  /pgf/arrow keys/.cd,
  ipe arrow normal/.style={scale=7},
  ipe arrow large/.style={scale=10},
  ipe arrow small/.style={scale=5},
  ipe arrow tiny/.style={scale=3},
  ipe arrow normal,
  /tikz/.cd,
  ipe arrows, %
  <->/.tip = ipe normal,
  ipe dash normal/.style={dash pattern=},
  ipe dash dotted/.style={dash pattern=on 1bp off 3bp},
  ipe dash dashed/.style={dash pattern=on 4bp off 4bp},
  ipe dash dash dotted/.style={dash pattern=on 4bp off 2bp on 1bp off 2bp},
  ipe dash dash dot dotted/.style={dash pattern=on 4bp off 2bp on 1bp off 2bp on 1bp off 2bp},
  ipe dash normal,
  ipe node/.append style={font=\normalsize},
  ipe stretch normal/.style={ipe node stretch=1},
  ipe stretch normal,
  ipe opacity 10/.style={opacity=0.1},
  ipe opacity 30/.style={opacity=0.3},
  ipe opacity 50/.style={opacity=0.5},
  ipe opacity 75/.style={opacity=0.75},
  ipe opacity opaque/.style={opacity=1},
  ipe opacity opaque,
]
\definecolor{red}{rgb}{1,0,0}
\definecolor{blue}{rgb}{0,0,1}
\definecolor{green}{rgb}{0,1,0}
\definecolor{yellow}{rgb}{1,1,0}
\definecolor{orange}{rgb}{1,0.647,0}
\definecolor{gold}{rgb}{1,0.843,0}
\definecolor{purple}{rgb}{0.627,0.125,0.941}
\definecolor{gray}{rgb}{0.745,0.745,0.745}
\definecolor{brown}{rgb}{0.647,0.165,0.165}
\definecolor{navy}{rgb}{0,0,0.502}
\definecolor{pink}{rgb}{1,0.753,0.796}
\definecolor{seagreen}{rgb}{0.18,0.545,0.341}
\definecolor{turquoise}{rgb}{0.251,0.878,0.816}
\definecolor{violet}{rgb}{0.933,0.51,0.933}
\definecolor{darkblue}{rgb}{0,0,0.545}
\definecolor{darkcyan}{rgb}{0,0.545,0.545}
\definecolor{darkgray}{rgb}{0.663,0.663,0.663}
\definecolor{darkgreen}{rgb}{0,0.392,0}
\definecolor{darkmagenta}{rgb}{0.545,0,0.545}
\definecolor{darkorange}{rgb}{1,0.549,0}
\definecolor{darkred}{rgb}{0.545,0,0}
\definecolor{lightblue}{rgb}{0.678,0.847,0.902}
\definecolor{lightcyan}{rgb}{0.878,1,1}
\definecolor{lightgray}{rgb}{0.827,0.827,0.827}
\definecolor{lightgreen}{rgb}{0.565,0.933,0.565}
\definecolor{lightyellow}{rgb}{1,1,0.878}
\definecolor{black}{rgb}{0,0,0}
\definecolor{white}{rgb}{1,1,1}
\begin{tikzpicture}[ipe stylesheet, xscale=0.7, yscale=0.7]
  \draw[shift={(159.516, 448.484)}, rotate=89.981]
    (0, 0) rectangle (32, -192);
  \draw[shift={(159.569, 576.484)}, rotate=89.981]
    (0, 0) rectangle (32, -192);
  \draw[shift={(255.516, 448.452)}, rotate=89.981]
    (0, 0)
     -- (32, 0);
  \draw[shift={(255.569, 576.452)}, rotate=89.981]
    (0, 0)
     -- (32, 0);
  \pic
     at (175.5215, 464.4786) {ipe disk};
  \pic
     at (207.5215, 464.468) {ipe disk};
  \pic
     at (239.5215, 464.4574) {ipe disk};
  \pic
     at (271.5215, 464.4467) {ipe disk};
  \pic
     at (303.5215, 464.4361) {ipe disk};
  \pic
     at (335.5215, 464.4255) {ipe disk};
  \pic
     at (335.5746, 592.4255) {ipe disk};
  \pic
     at (303.5746, 592.4361) {ipe disk};
  \pic
     at (271.5746, 592.4467) {ipe disk};
  \pic
     at (239.5746, 592.4574) {ipe disk};
  \pic
     at (207.5746, 592.468) {ipe disk};
  \pic
     at (175.5746, 592.4786) {ipe disk};
  \draw[red, ipe pen heavier]
    (176, 464)
     -- (207.574, 592.468)
     -- (207.574, 592.468);
  \draw[red, ipe pen heavier]
    (176, 464)
     -- (239.574, 592.458)
     -- (239.574, 592.458);
  \draw[red, ipe pen heavier]
    (176, 464)
     -- (303.574, 592.437)
     -- (303.574, 592.437);
  \draw[red, ipe pen heavier]
    (176, 464)
     -- (335.574, 592.426)
     -- (335.574, 592.426);
  \draw
    (208, 464)
     -- (175.574, 592.479)
     -- (175.574, 592.479);
  \draw
    (208, 464)
     -- (239.574, 592.458)
     -- (239.574, 592.458);
  \draw
    (208, 464)
     -- (271.574, 592.447)
     -- (271.574, 592.447);
  \draw
    (208, 464)
     -- (335.574, 592.426)
     -- (335.574, 592.426);
  \draw
    (240, 464)
     -- (175.574, 592.479)
     -- (175.574, 592.479);
  \draw
    (240, 464)
     -- (207.574, 592.468)
     -- (207.574, 592.468);
  \draw
    (240, 464)
     -- (271.574, 592.447)
     -- (271.574, 592.447);
  \draw
    (240, 464)
     -- (303.574, 592.437)
     -- (303.574, 592.437);
  \draw
    (272, 464)
     -- (207.574, 592.468)
     -- (207.574, 592.468);
  \draw
    (272, 464)
     -- (239.574, 592.458)
     -- (239.574, 592.458);
  \draw
    (272, 464)
     -- (303.574, 592.437)
     -- (303.574, 592.437);
  \draw
    (272, 464)
     -- (335.574, 592.426)
     -- (335.574, 592.426);
  \draw
    (304, 464)
     -- (175.574, 592.479)
     -- (175.574, 592.479);
  \draw
    (304, 464)
     -- (239.574, 592.458)
     -- (239.574, 592.458);
  \draw
    (304, 464)
     -- (271.574, 592.447)
     -- (271.574, 592.447);
  \draw
    (304, 464)
     -- (335.574, 592.426)
     -- (335.574, 592.426);
  \draw[red, ipe pen heavier]
    (336, 464)
     -- (175.574, 592.479)
     -- (175.574, 592.479);
  \draw[red, ipe pen heavier]
    (336, 464)
     -- (207.574, 592.468)
     -- (207.574, 592.468);
  \draw[red, ipe pen heavier]
    (336, 464)
     -- (271.574, 592.447)
     -- (271.574, 592.447);
  \draw[red, ipe pen heavier]
    (336, 464)
     -- (303.574, 592.437)
     -- (303.574, 592.437);
  \node[ipe node, font=\large]
     at (203.498, 432.468) {$A_{2p-1}$};
  \node[ipe node, font=\large]
     at (287.498, 432.437) {$A_{2p}
$};
  \node[ipe node, font=\large]
     at (291.591, 616.437) {$B_{2p}$};
  \node[ipe node, font=\large]
     at (203.591, 616.466) {$B_{2p-1}$};
  \node[ipe node, font=\Large]
     at (171.508, 452.479) {$u$};
  \node[ipe node, font=\Large]
     at (203.508, 452.468) {$v$};
  \node[ipe node, font=\Large]
     at (235.508, 452.458) {$w$};
  \node[ipe node, font=\Large]
     at (267.508, 452.479) {$u$};
  \node[ipe node, font=\Large]
     at (299.508, 452.468) {$v$};
  \node[ipe node, font=\Large]
     at (331.508, 452.458) {$w$};
  \node[ipe node, font=\Large]
     at (171.508, 596.479) {$u$};
  \node[ipe node, font=\Large]
     at (203.508, 596.468) {$v$};
  \node[ipe node, font=\Large]
     at (235.508, 596.457) {$w$};
  \node[ipe node, font=\Large]
     at (267.508, 596.479) {$u$};
  \node[ipe node, font=\Large]
     at (299.508, 596.468) {$v$};
  \node[ipe node, font=\Large]
     at (331.508, 596.457) {$w$};
  \draw[red]
    (172, 468) rectangle (180, 460);
  \draw[red]
    (332, 468) rectangle (340, 460);
\end{tikzpicture}